\documentclass[acmsmall,screen]{acmart}

\startPage{1}

\usepackage{xcolor}
\usepackage{comment}
\usepackage{amsmath}
\usepackage{hyperref}
\usepackage{tikz}
\usepackage{xspace}

\usepackage{makecell}

\usepackage{csquotes}

\usepackage{algorithm}
\usepackage[indLines=false, rightComments=false, beginComment=//]{algpseudocodex}

\usepackage{amsthm}

\theoremstyle{definition}
\newtheorem{definition}{Definition}[section]

\newtheorem{theorem}{Theorem}[section]
\newtheorem{lemma}[theorem]{Lemma}

\usepackage{mathtools}
\usepackage{stmaryrd}
\usepackage{dsfont}

\usepackage{subcaption}
\usepackage{tablefootnote}

\usepackage{mathpartir}

\newcommand\bind{\gg\kern-.5ex=\kern-.3ex}

\newcommand{\tighten}[1]{{\,{#1}\,}}

\newcommand\hastype{\tighten{:}}

\newcommand\nat{\mathbb{N}}
\newcommand\bool{\mathbb{B}}
\newcommand\unitt{\mathds{1}}
\newcommand\unitv{\mathtt{tt}}
\newcommand{\option}[1]{{\mathsf{option}\;{#1}}}
\newcommand{\Some}[1]{{\mathsf{Some}\;{#1}}}
\newcommand\None{\mathsf{None}}
\newcommand\funarrow{\tighten{\rightarrow}}

\DeclarePairedDelimiter\ceil{\lceil}{\rceil}

\newcommand{\pure}[2]{{\ceil{#1}_{#2}}}
\newcommand{\communicate}[4]{\texttt{comm}\,\,#1\,\,#2\,\,#3\,\,#4}
\newcommand{\letexp}[3]{\texttt{let}\,\,#1\,\coloneqq\,#2\,\texttt{in}\,\,#3}

\newcommand\doublebackslash{\tighten{\backslash\kern-1.5ex\backslash\kern0.4ex}}

\newcommand{\Pseudocomment}[1]{{\small\textit{#1}}}

\newcommand{\iter}[2]{{#1}^{#2}}

\newcommand{\vectortp}[2]{\mathsf{Vec}\,{#1}\,{#2}}

\newcommand{\texp}[4]{{#1} \vdash_{#2} {#3} \hastype #4}
\newcommand{\tprog}[5]{{#1};{#2} \vdash_{#3} {#4} \hastype {#5}}

\newcommand{\linen}[1]{\color{gray}{#1}\ }
\newcommand{\myifthenelse}[3]{\mathsf{if}\;{#1}\;\mathsf{then}\;{#2}\;\mathsf{else}\;{#3}}

\newcommand\doubleplus{+\kern-1.3ex+\kern0.8ex}
\newcommand\cons{::}

\newcommand{\denote}[1]{{\llbracket #1 \rrbracket}}
\newcommand{\compile}[2]{{\llparenthesis \,{#1}\, \rrparenthesis_{#2}}}
\newcommand{\powerset}[1]{{\mathcal{P}({#1})}}
\newcommand{\netsim}{\mathsf{Netwk}}
\newcommand{\Dom}[1]{\mathsf{Dom}(#1)}

\newcommand{\alignl}[2]{{\mathsf{align}_{#1}({#2})}}
\newcommand{\projl}[2]{{\mathsf{proj}_{#1}({#2})}}
\newcommand{\filterl}[2]{{\mathsf{filter}_{#1}(#2)}}
\newcommand{\mapl}[2]{{\mathsf{map}\;{#1}\;{#2}}}
\newcommand{\foldl}[3]{{\mathsf{foldl}\;{#1}\;{#2}\;{#3}}}
\newcommand{\flatmapl}[2]{{\mathsf{flatmap}\;{#1}\;{#2}}}
\newcommand{\orderrelated}{{\sim}}
\newcommand{\extract}[1]{{#1}{\downarrow}}

\newcommand\HLL{Sync\xspace}
\newcommand\LLL{Async\xspace}

\newcommand{\Lrole}{L}
\newcommand{\Rrole}{R}

\begin{document}

    \title{Functional Reasoning for Distributed Systems with Failures}

    \author{Haobin Ni}
    \affiliation{\institution{Cornell University} \city{Ithaca, NY} \country{USA}}

    \author{Robbert van Renesse}
    \affiliation{\institution{Cornell University} \city{Ithaca, NY} \country{USA}}

    \author{Greg Morrisett}
    \affiliation{\institution{Cornell University} \city{Ithaca, NY} \country{USA}}

    \begin{abstract}

    Distributed system theory literature often argues for correctness using an informal, Hoare-like style of reasoning.
    While these arguments are intuitive, they have not all been foolproof, and whether they directly correspond to formal proofs is in question.
    
    We formally ground this kind of reasoning and connect it to standard formal approaches through language design and meta-analysis, which leads to a functional style of compositional formal reasoning for a class of distributed systems, including cases involving Byzantine faults.
    
    The core of our approach is twin languages: \HLL and \LLL, which formalize the insight from distributed system theory that an asynchronous system can be reduced to a synchronous system for more straightforward reasoning under certain conditions. 
    \HLL describes a distributed system as a single, synchronous, data-parallel program.
    It restricts programs syntactically and has a functional denotational semantics suitable for Hoare-style formal reasoning. 
    \LLL models a distributed system as a collection of interacting monadic programs, one for each non-faulty node in the system.
    It has a standard trace-based operational semantics, modeling asynchrony with interleaving.
    \HLL compiles to \LLL and can then be extracted to yield executable code. 
    We prove that any safety property proven for a \HLL program in its denotational semantics is preserved in the operational semantics of its compiled \LLL programs.
    We implement the twin languages in Rocq and verify the safety properties of two fault-tolerant consensus protocols: BOSCO and SeqPaxos.
    
    \end{abstract}
    
    \maketitle
    
    \section{Introduction}

    Distributed system theory often argues for correctness using an informal, Hoare-like style of reasoning ~\cite{lamportParttimeParliament1998, songBoscoOneStepByzantine2008, kotlaZyzzyvaSpeculativeByzantine2010, ongaroSearchUnderstandableConsensus2014, abrahamSyncHotStuffSimple2019, buterinCombiningGHOSTCasper2020}. 
    These arguments are intuitive but are sometimes unsound ~\cite{abrahamRevisitingFastPractical2017, michaelRecoveringSharedObjects2017, momoseForceLockingAttackSync2019, neuEbbandFlowProtocolsResolution2021}.
    Their formal implications are also unclear, as the primary approaches to formally reasoning about distributed systems have been trace-based operational semantics and global invariants to handle asynchrony and faults ~\cite{chandyDistributedSnapshotsDetermining1985, alpernRecognizingSafetyLiveness1987, lamportSpecifyingSystemsTLA2002, wilcoxVerdiFrameworkImplementing2015, hawblitzelIronFleetProvingPractical2015, maI4IncrementalInference2019, qiuLiDODAGFrameworkVerifying2025}.

    In this paper, we formally ground this informal style of reasoning and connect it to the standard approach through language design and meta-analysis.
    In doing so, we establish a functional denotational semantics for distributed systems, enabling a class of asynchronous systems, potentially with Byzantine faults, to be reasoned compositionally following the syntax of their programs, as functions with non-deterministic outputs.

    We propose twin languages, \HLL and \LLL, inspired by distributed system theory ~\cite{charron-bostHeardOfModelComputing2009} that an asynchronous system can sometimes be reduced to a synchronous system for reasoning. 
    \HLL enables the aforementioned style of reasoning for a class of distributed systems, as if they were synchronous and data-parallel. 
    \LLL and the relationship between the twin languages ensure that any safety property proven with \HLL's functional semantics remains true when the system executes asynchronously as described by \LLL's standard trace semantics. 

    \begin{figure}[t]
        \centering
        \hspace*{-.25in}
        \begin{subfigure}[t]{0.5\textwidth}
            \centering
            \begin{tabular}{l}
                $\mathtt{SimpleVote}\ (p \hastype \bool@\Lrole, x \hastype \bool@\Rrole)$
                \\ $ \quad \quad \hastype (\texttt{option}\;\bool)@\Lrole \coloneqq$  \\
                $\quad \textsf{let}\ \mathtt{cnt} \coloneqq \texttt{comm}\ c\ x\ 0\ (\mathtt{fcnteq}\ p)\ \textsf{in}$ \\
                $\quad \texttt{ret}\ (\mathtt{calc\_dec}\ \mathtt{cnt}\ p)$ 
            \end{tabular}
            \vspace{-0.02in}
            \caption{\texttt{SimpleVote} in simplified \HLL syntax}
            \label{fig:example-hll}
        \end{subfigure}
        ~
        \hspace*{-.15in}
        \begin{subfigure}[t]{0.55\textwidth}
            \centering
            \begin{tabular}{l}
                $\mathtt{fcnteq}\ (p \hastype \bool) (\mathtt{cnt} \hastype \nat) (v \hastype \bool) \hastype \nat \coloneqq $
                \\
                $\quad \mathsf{if}\ (p \tighten{==} v)\ \mathsf{then}\ (\mathtt{cnt}+1)\ \mathsf{else}\ \mathtt{cnt}$\\
                $\mathtt{calc\_dec}\ (\mathtt{cnt} \hastype \nat)(p \hastype \bool)\hastype \mathtt{option}\;\bool \coloneqq$ \\
                $\quad \mathsf{if}\ (\mathtt{cnt} \ge n \tighten{-} 2 * f)\ \mathsf{then}\ (\Some{p})\ \mathsf{else}\ \None$
            \end{tabular}
            \vspace{-0.02in}
            \caption{Auxiliary pure function definitions}
            \label{fig:example-fun}            
        \end{subfigure} 
        \Description[]{}
        \vspace{-.15in}
    \end{figure}

    We illustrate the problem and our contribution by defining and reasoning about an example program.
    Figure \ref{fig:example-hll} shows the \texttt{SimpleVote} program in simplified \HLL syntax.
    The \HLL syntax is \textit{choreographic}, describing the whole system with a single program. 
    Two auxiliary functions are defined separately as pure functions in Figure \ref{fig:example-fun}.
    
    \texttt{SimpleVote} aims to describe a distributed system comprising a single leader node and multiple replica nodes that collectively decide on a boolean value based on their local inputs.
    The program has three steps: first, all replica nodes send their local inputs to the leader; second, the leader node waits until it receives the enough messages from replicas and counts how many of the received values are equal to its local input with $\mathtt{fcnteq}$; finally, the leader node computes its local output by checking if the number of equal values reaches a threshold with $\mathtt{calc\_dec}$.

    More concretely, \texttt{SimpleVote} assumes two \textit{roles}, i.e., classes of nodes that execute the same code locally: leader and replica, denoted as $\Lrole$ and $\Rrole$ in the syntax.
    It has two implicit parameters, $n$ and $f$, assuming there is a single non-faulty leader node and $n$ replica nodes, with at most $f$ replica nodes being Byzantine and behaving arbitrarily.
    The first line declares that the program takes in two \emph{distributed} values $p$ and $x$. 
    $p$ is a single Boolean input to the leader, and $x$ contains a Boolean input for each of the replica nodes.
    The second line specifies that the return value is an optional Boolean value at the leader node.
    The following line is a let binding that binds the result of a \emph{communication} action to variable $\mathtt{cnt}$.
    The communication action specifies replica nodes should send their input value $x$ to the leader node, who would wait for at least $n - f$ messages and count the number of messages containing a value equal to $p$ by folding the function $(\mathtt{fcnteq}\ p)$ over the list of received messages with the default value $0$.
    The name $c$ in the line uniquely identifies this particular round of communication.
    The final line specifies that the leader should compute $\mathtt{calc\_dec}\ \mathtt{cnt}\ p$ to produce an output. 
    In particular, if there are at least $n - 2f$ votes containing $p$, the leader should return $\Some{p}$ and $\None$ otherwise.

    A safety property of this program is that when $n > 3f$ and the leader and all correct replicas have the same Boolean value $b$ as their input, the output of the leader must be $\Some{b}$, no matter the actions of the faulty nodes.

    We argue informally in three steps following the syntax of \texttt{SimpleVote}.
    First, because the leader waits for at least $n - f$ messages, among which at most $f$ messages come from faulty nodes, there are at least $n - 2f$ messages received from the correct replica nodes, and all of those contain the value $b$.
    Second, by the definition $\mathtt{fcnteq}$, this implies $\mathtt{cnt}$ is at least $n - 2f$.
    Last, by the definition of $\mathtt{calc\_dec}$, we derive that the leader's output is $\Some{b}$.

    As we can see, this style of reasoning is more straightforward than devising global invariants that would imply the safety property to be proven, which must cover all possibilities of asynchrony.
    However, could this be too good to be true?
    Two questions left to answer are:
    
    \begin{enumerate}
        \item How can we make these arguments formally and precisely?
        \item Do we get the same guarantees as if we had done reasoning using invariants on traces? 
    \end{enumerate}

    \textbf{Our first contribution is to formally ground this style of high-level reasoning.} 
    We address (1) by designing \HLL and its functional denotational semantics, allowing arguments in the above style to be stated precisely as formal proofs.

    \textbf{Our second contribution is to formally connect to established verification paradigms.} 
    We address (2) by proving the adequacy of \HLL's semantics for reasoning in regard to a canonical asynchronous model of distributed systems, \LLL.
    \LLL is a low-level language with a standard trace-based operational semantics that models asynchrony as interleaving, which can be reasoned with standard formal verification techniques.
    More specifically, we prove that for any program $p$ in \HLL, any output produced by the traces derived from compiling $p$ to \LLL, is included in the functional semantics of \HLL.
    This guarantees that any safety properties proven in the \HLL semantics also hold for the compilation result in the \LLL semantics.

    \textbf{Our third contribution is to present a novel framework of formal functional reasoning for the implementation of a class of distributed systems with failures.}
    By combining \HLL and \LLL, one can reason about a distributed system as a composition of non-deterministic functions following its program syntax, while providing strong guarantees about the behavior of the extracted executable code.
    To showcase this, we have implemented the twin languages in the Rocq proof assistant and verified two consensus protocols as case studies: Bosco, a one-step Byzantine-fault tolerant consensus protocol;
    and Sequential Paxos, a crash-fault tolerant consensus protocol that is a variant of Paxos~\cite{lamportParttimeParliament1998}.
    Our more precise reasoning uncovered optimizations that were missed in the original protocols.
    
    For the rest of the paper,
    Section \ref{sec:HLL} presents the syntax and semantics of \HLL, and Section \ref{sec:case_studies} applies it to verify the two consensus protocols.
    Section \ref{sec:operational} clarifies our assumptions about the network behavior and adversarial node and presents the syntax and semantics of \LLL, and Section \ref{sec:proof} proves the relationship between the twin languages.    

    \section{Syntax and High-Level Semantics for \HLL}
    \label{sec:HLL}

    This section introduces \HLL, our high-level language for modeling Byzantine fault-tolerant distributed systems.  
    \HLL is in choreography style \cite{montestiChoreographicProgramming2013, zdancewicSecureProgramPartitioning2002}, which means a \HLL program describes the whole distributed system and can be compiled into programs for individual nodes through \emph{endpoint projection} to be defined later in Section \ref{sec:operational}.
    
    \HLL puts three restrictions on the set of expressible distributed systems.
    First, \HLL does not support branches or loops (except for a meta-level loop over the whole program) due to the problem known as \emph{knowledge of choice} (KoC)~\cite{castagnaGlobalTypesMultiparty2011}, where some nodes may not know which branch to take.
    This becomes sort of a chicken-and-egg situation with failures, as consensus protocols are often used to resolve KoC in a fault-tolerant way in practice.
    Second, \HLL cannot model various cryptographic primitives, so distributed systems that depend on them cannot be expressed.
    How to formally and modularly reason about distributed systems relying on more sophisticated usage of cryptography is an open problem on its own.
    We left supporting specific primitives for future work.
    Last, we enforce a total order on the communication actions in the system.
    As a result, this version of \HLL only supports essentially streamlined programs, whose \LLL semantics satisfy the definition of \emph{communication-closed} programs ~\cite{damianCommunicationClosedAsynchronousProtocols2019}.
    
    At a high level, the denotational semantics of a closed \HLL program is a \emph{set} of possible distributed value combinations on some set of nodes.
    The denotational semantics of a \HLL program with free distributed variables is a function that transforms input values for those variables to a set of possible output distributed value combinations.
 
    \subsection{\HLL Syntax}
    \label{sec:HLL_syntax}

    Our \HLL syntax uses the \emph{role} abstraction to describe distributed systems with different kinds of nodes (such as a leader or a replica) while abstracting the number of nodes that will be executing the code associated with a given role.  
    We use $\Lrole$ and $\Rrole$ as example roles.  

    A \HLL program describes an abstract distributed protocol as a high-level functional program that runs the same code for each node in a given role, albeit on different input values.
    Its syntax is defined as follows:

    \begin{definition}[\HLL Syntax]
        \begin{equation*}
            \begin{array}{lr}
                \text{Roles} & \Rrole, \Lrole \\
                \text{Meta-terms} & t \\
                \text{Channels} & c \\
                \text{Distributed Variables}& x \\
            \end{array}
        \end{equation*}
        \begin{equation*}
            \begin{array}{rcll}
                \text{\HLL Expressions} &&& \\
                    e & \Coloneqq & x_{\Rrole} & \text{(role-local variables)} \\
                      &           & \pure{t}{\Rrole} & \text{(replicated Rocq terms)} \\
                      &           & [t_1,\ldots,t_n]_{\Rrole} & \text{(vector of Rocq terms)} \\
                      &           & e_1\,e_2 & \text{(application)} \\
                                   \\
                 \text{\HLL Programs} &&& \\
                    p & \Coloneqq & \texttt{ret}\,\{R_1\mapsto e_1,\ldots,R_n\mapsto e_n\} & \text{(return a record)}\\
                      &           & \letexp{x}{p_1}{p_2} & \text{(sequencing)}\\
                      &           & \communicate{c}{e_m}{e_d}{e_f} & {\text{(communication)}}
            \end{array}            
        \end{equation*}
    \end{definition}

    \HLL is realized via a shallow embedding that inherits all of Rocq's terms at the expression level and adds \emph{programs} ($p$) that operate over a record mapping roles to vectors of values, one vector element for each node in the given role.  Critically, the only way for nodes to interact with each other is through communication.  Thus, a \HLL program $p$ is essentially a sequence of communication steps, where sequencing is accomplished through $\texttt{let}$ and terminated with a $\texttt{ret}$.  The $\texttt{ret}$ operation takes a record of expressions, one for each role, and those expressions are calculated for each node within that role, by mapping the expression's denotation over a vector of environments, one for each node. Thus, the variables bound in a \HLL $\texttt{let}$ are records mapping roles to vectors of values.  In order to ensure that one node cannot access values from another node, $\texttt{let}$-bound variables are restricted within expressions:  a given role can only access its vector of the record, and implicitly, a given node can only access its element of the vector.  In our Rocq formalism, we use parametric higher-order abstract syntax (PHOAS)~\cite{chlipalaParametricHigherorderAbstract2008} to represent variables and binding and enforce these constraints.  

    The program $\communicate{c}{e_m}{e_d}{e_f}$ involves a channel $c$, a sender role (given by the type of $e_m$), and a receiver role (given by the types of $e_d$ and $e_f$).
    When executed, all of the sender nodes calculate a message $e_m$ and send that message's value to all of the nodes in the receiver role via the channel $c$.  The receiver role's nodes each run a message handler that accepts the incoming messages and combines them with a folding function $e_f$ and a default value $e_d$. The channel  $c$ is used by the sender and the receiver nodes to distinguish messages for different rounds of communication. Our type system enforces that channel names are unique and used in a particular order. At the implementation level, this uniqueness is realized through sequence numbers.

    \subsection{Typing for \HLL Programs}
    \label{sec:HLL_typing}

    Our typing rules are largely standard but ensure that each expression is situated for a given role and each channel $c$ is used exactly once in a specified order.
    The typing judgment for expressions has the form $\texp{\Gamma}{\Rrole}{e}{\tau}$ where $\Gamma$ is a context mapping variables to record types indexed by roles, and is defined as follows:
    \begin{mathpar}
        \infer*[left=Var]
        { \Gamma(x) = \{\Rrole_1\hastype \tau_1,\ldots,\Rrole_n\hastype \tau_n\}}
        { \texp{\Gamma}{\Rrole_i}{x_{\Rrole_i}}{\tau_i}} \;\;\;
        
        \infer*[left=Lift]
        { t\hastype\tau }
        { \texp{\Gamma}{\Rrole}{\pure{t}{\Rrole}}{\tau}} \\

        \infer*[left=Vector]
        { t_1\hastype\tau \cdots t_n\hastype\tau}
        { \texp{\Gamma}{\Rrole}{[t_1,\ldots,t_n]}{\tau} } \;\;\;
        
        \infer*[left=App]
        { \texp{\Gamma}{\Rrole}{e_1}{\tau' \rightarrow \tau} \qquad \texp{\Gamma}{\Rrole}{e_2}{\tau'} }
        { \texp{\Gamma}{\Rrole}{e_1\,e_2}{\tau}}
    \end{mathpar}
    For variables, we find the record type associated with the variable by the context $\Gamma$, and then extract the type associated with the given role.  Note that it is not necessary for a given role to be present for all variables, so $x_{\Rrole}$ is only well-formed when $x$ is in $\Dom{\Gamma}$ and $\Rrole$ is in $\Dom{\Gamma(x)}$.
    For inherited Rocq terms, we simply ascribe them the type that Rocq gives them, but situated at a given role.  Finally, application is typed as expected.  
    
    The typing judgment for programs has the form $\tprog{\Delta}{\Gamma}{\mathcal{R}}{p}{\{R_i \hastype \tau_i\}}$ where $\Delta$ is an \emph{ordered} channel context mapping channel names $c$ to a triple $(S,R,\tau)$ of a sending role, a receiving role, and a type for messages to be sent on the channel, and where $\mathcal{R}$ is a set of roles that can be used in the program.  The judgment is defined with the following rules:
    \begin{mathpar}
        \infer*[left=Ret]
        { \texp{\Gamma}{\Rrole_i}{e_i}{\tau_i} \quad \Rrole_i \in \mathcal{R} }
        { \tprog{\emptyset}{\Gamma}{\mathcal{R}}{\mathtt{ret}\,\{\Rrole_i \mapsto e_i\}}{\{\Rrole_i \hastype \tau_i\}} } 

        \infer*[left=Let]
        { \tprog{\Delta_1}{\Gamma}{\mathcal{R}}{p_1}{\tau_1} \quad \tprog{\Delta_2}{\Gamma[x\hastype \tau_1]}{\mathcal{R}}{p_2}{\tau_2}\quad \Dom{\Delta_1}\cap \Dom{\Delta_2} = \emptyset }
        { \tprog{\Delta_1 \doubleplus \Delta_2}{\Gamma}{\mathcal{R}}{\letexp{x}{p_1}{p_2}}{\tau_2} }

        \infer*[left=Comm]
        { \texp{\Gamma}{\{S\}}{e_m}{\tau_m} \quad \texp{\Gamma}{\{R\}}{e_d}{\tau} \quad \texp{\Gamma}{\{R\}}{e_f}{\tau \rightarrow \tau_m \rightarrow \tau} \quad S,R \in \mathcal{R}}
        { \tprog{[c : (S,R,\tau_m)]}{\Gamma}{\mathcal{R}}{\communicate{c}{e_m}{e_d}{e_f}}{\{R \hastype \tau\}} }
    \end{mathpar}
    For $\mathtt{ret}$, we simply check that each expression in the returned record is well-typed at the corresponding role.  In this case, the channel context is empty since this command does not do any communication.  For $\letexp{x}{p_1}{p_2}$, we check that $p_1$ has a (record) type $\tau_1$ and then extend $\Gamma$ with the assumption that $x\hastype\tau_1$ to check that $p_2$ has the (record) type $\tau_2$.  Note that here, the channel contexts are checked to be disjoint and are appended in order.  For $\communicate{c}{e_m}{e_d}{e_f}$, the channel context has \emph{only} $c$ associated with a sending role $S$, receiving role $R$, and message type $\tau_m$.  We check that the message expression ($e_m$) is typed at the sending role $S$ with type $\tau_m$, and that the default ($e_d$) and combining function ($e_f$) expressions are typed at the receiving role $R$ as $\tau$ and $\tau\rightarrow\tau_m\rightarrow\tau$ respectively. 
    Finally, the communication step only returns $\tau$ values for the nodes of the receiver role, so the resulting record type is the singleton $\{ R \hastype \tau \}$.
    
    \subsection{Denotational Semantics for \HLL}

    \subsubsection{Notation:}
    
    For any natural number $n$ and type $\tau$, $\vectortp{n}{\tau}$ describes a sequence of $n$ elements of type $\tau$.
    We use $\vec{v}$ to make clear the value is a vector value and write $\vec{v}@i$ for the operation that projects element $i$ from vector $\vec{v}$.  
    The operations $\mathsf{map}\ f\ \vec{v}$, $\mathsf{foldl}\ f\ d\ \vec{v}$, $\mathsf{zip}\ \vec{v}_1 \ \vec{v}_2$, $\mathsf{unzip}\ \vec{v}$ are the usual map, fold, zip, and unzip for vectors.
    The operation $\mathsf{const}\,n\,v$ replicates the value $v$ $n$ times to produce a vector.  
    
    For any set or type $\tau$, $\powerset{\tau}$ denotes the power set of $\tau$.
    We also use the power set as a monad with the following notations.
    \begin{equation*}
        \begin{array}{ll}
            \powerset{f} \hastype \powerset{A} \tighten{\rightarrow} \powerset{B} & \makecell[l]{\text{lifting function } f \text{ into its set version by applying } f \text{ to all elements}} \\
            \mathsf{singleton}(x) = \{x\} \hastype \powerset{\tau} & \text{the singleton set, also the monadic return} \\
            \bind \hastype \powerset{A} \tighten{\rightarrow} (A \tighten{\rightarrow} \powerset{B}) & \makecell[l]{ \text{the bind applies the function to all elements and takes the union}}
        \end{array}
    \end{equation*}
    For any record $e$, we write $e@R$ to extract the record field named by $R$. We also have a shorthand $\Pi \hastype \vectortp{n}{(\powerset{\tau})} \tighten{\rightarrow} \powerset{\vectortp{n}{\tau}}$ that forms a set of vectors from a vector of sets by enumerating all combinations of each element.
    
    \subsubsection{Configuration}
    
    The denotation of a \HLL program is parameterized by a \emph{configuration}, which is defined as follows:
    
    \begin{definition}[Distributed System Configuration]
        \label{def:configuration}
        Given a role set $\mathcal{R}$, a configuration $\mathcal{C}$ defines the following for each role $\Rrole \in \mathcal{R}$:
        \begin{itemize}
            \item The total number of nodes $n_\Rrole$ and an upper bound on the number of faults to tolerate $f_\Rrole$.

            \item A set of $g_\Rrole$ nodes that are correct or follow the protocol until crash $\mathsf{Good}_{\Rrole} = \{r_1, \dots, r_{g_\Rrole}\}$.

            \item A set of $b_\Rrole$ Byzantine nodes $\mathsf{Byz}_{\Rrole}$, $b_\Rrole \le f_\Rrole$ and $b_\Rrole + g_\Rrole = n_\Rrole$.

        \end{itemize}
    \end{definition}
    Given a configuration, we define a relation $\netsim_{\Rrole}$ between a vector of messages sent by nodes in $\mathsf{Good}_{\Rrole}$ to a list of messages that a receiver node may possibly receive.  We use this relation to calculate all of the ``bad'' things that can happen in the network, such as dropping a message, permuting the order in which messages are received, or injecting Byzantine messages.  In our setting, $\netsim_{\Rrole}$ is defined as a predicate (\texttt{prop}) on a vector and list of messages as follows:
    \begin{equation*}
        \netsim_{\Rrole} \coloneqq \lambda (\vec{v} \hastype \vectortp{g_\Rrole}{\tau}).\ (\mathsf{add\_any}\, \vec{v}\, b_\Rrole) \bind \mathsf{perm} \bind (\lambda \vec{v}.\ \mathsf{trunc}\, \vec{v}\, (n_\Rrole - f_\Rrole))
    \end{equation*}
    The term $\mathsf{add\_any}\, \vec{v}\, b_\Rrole$ produces the set of all vectors that add up to $b_\Rrole$ arbitrary values of type $\tau$ to the end of $\vec{v}$. The added messages model those sent by a Byzantine node and can take any value. Note that Byzantine nodes may also send multiple messages, but we assume a correct receiver will only process at most one message from each sender node.  The term $\mathsf{perm}$ produces the set of all permutations of the input vector and models the arbitrary order in which messages may be received.
    Finally, the term $\mathsf{trunc}\, \vec{v}\, (n_\Rrole - f_\Rrole)$ produces all prefixes of $\vec{v}$ whose length is at least $n_\Rrole - f_\Rrole$.
    This models a receiver node not receiving some of the messages.  
    Due to the asynchronous network, a node cannot distinguish between a message being delayed and a message never sent.
    So in practice, a receiver only waits for a certain number of messages before continuing.
    In our case, the maximum number of messages to wait for is $n_\Rrole - f_\Rrole$ because a correct node always sends its message, and the total number of faulty nodes, either crashing or Byzantine, is bounded by $f_\Rrole$. 

    This definition is an over-approximation of the possible behaviors we can observe.  For instance, this definition allows Byzantine nodes to generate messages with ``secrets'' that they may not know.  Nevertheless, the over-approximation ensures that we cover all of the cases that need to be considered and is yet precise enough that we can still prove useful properties.  
    
    \subsubsection{Denotational Semantics}

    We begin by giving a denotation to \HLL program types and in particular, we define:
    \[
        \denote{\{R_1\hastype\tau_1,\ldots,R_n\hastype\tau_n\}} = \{R_1\hastype\vectortp{g_{R_1}}{\tau_1},\ldots,R_n\hastype\vectortp{g_{R_n}}{\tau_n} \}
    \]
    
    That is, $\HLL$ records are translated to records of vectors, where the lengths of the vectors are determined by the parameters of the configuration.  Next, we lift the translation to variable contexts by defining:
    \[
        \denote{[x_1\hastype\tau_1,\ldots,x_n\hastype\tau_n]} = \{x_1\hastype\denote{\tau_1},\ldots,x_n\hastype\denote{\tau_n} \}
    \]
    
    Thus, environments are represented as records mapping variable names to values, where the values are records mapping roles to vectors.  Now we can define the denotation of (derivations of) terms as a recursively defined meta-function with type:
    \[
        \denote{\tprog{\Delta}{\Gamma}{\mathcal{R}}{p}{\tau}} : \denote{\Gamma} \rightarrow \powerset{\denote{\tau}}
    \]
    which is defined as follows:\footnote{Formally, the definition is over derivations of the typing judgment, but to simplify the presentation, we present the definition over \HLL syntax}:

    \begin{definition}[Denotational Semantics of \HLL Programs]
        \label{def:denotational_sem}
        \begin{align*}
        \denote{\mathtt{ret}\,\{R_i\mapsto e_i\}} = \lambda v.\ & \mathsf{singleton}\ \{ R_i \mapsto \denote{e_i}v \} \\
        \denote{\letexp{x}{p_1}{p_2}} = \lambda v.\ & \denote{p_1}v \bind (\lambda y.\ \denote{p_2}(v[x\mapsto y])) \\
        \denote{\communicate{c}{e_m}{e_d}{e_f}} = \lambda v.\ & \mathsf{let}\ \mathtt{msgs} \coloneqq \denote{e_m}v\ \mathsf{in} \\
                                                              & \mathsf{let}\ \mathtt{netmsgs} \coloneqq \netsim_{S}(\mathtt{msgs}) \ \mathsf{in}\\ 
                                                              & \mathsf{let}\ \mathtt{pairs} \coloneqq \mathsf{zip} (\denote{e_f}v)\ (\denote{e_d}v) \ \mathsf{in}\\
                                                              & \mathsf{let}\ \mathtt{app} \coloneqq \lambda (f,d).\ \powerset{\mathsf{foldl}\ f\ d}\ \mathtt{netmsgs} \ \mathsf{in} \\
                                                              & \powerset{\lambda x.\{R\mapsto x\}} (\Pi(\mathsf{map}\ \mathtt{app}\ \mathtt{pairs}))
        \end{align*}
    \end{definition}
    The definition relies upon a denotation for \HLL expressions which has type:
    \[
      \denote{\texp{\Gamma}{R}{e}{\tau}} : \denote{\Gamma} \rightarrow \vectortp{g_R}{\tau}
    \]
    and is defined by:
    \begin{align*}
        \denote{x@R} = \lambda v.\ & v@x@R \\
        \denote{\pure{t}{R}} = \lambda v.\ & \mathsf{const}\ g_R\ t \\
        \denote{[t_1,\ldots,t_n]_{R}} = \lambda v.\ &[t_1,\ldots,t_n] \quad (n = g_R) \\
        \denote{e_1\ e_2} = \lambda v. \ & \mathsf{map}\ (\lambda (f,x).f\ x)\ (\mathsf{zip}\ \denote{e_1}v\ \denote{e_2}v) 
    \end{align*}

    Working backwards, the denotation of a variable simply extracts that variable from the environment ($v@x$), which yields a record of vectors. We then extract the vector of the corresponding role $R$ at which the expression is situated.  The denotation of an embedded Rocq term $t$ replicates the term as a vector for each of the good nodes at the given role. 
    The denotation of a vector of Rocq terms is just that vector, provided its length agrees with the number of nodes in the given role.  In practice, we only use vector literals in the meta-theory, so this assumption is easy to discharge.
    The denotation of an application $e_1\,e_2$ is given by calculating the denotation of $e_1$ and $e_2$ respectively, which should return a vector of functions and a vector of values, respectively.  We zip the two vectors into a vector of pairs, and then map the apply function across the resulting vector.  

    The denotation of programs is more involved because we return a set of records of vectors. We use the powerset monad to make the definitions a little simpler.  The cases for $\mathtt{ret}$ and $\mathtt{let}$ straightforwardly translate into the singleton and bind in the powerset monad.  For $\communicate{c}{e_m}{e_d}{e_f}$, we first compute a vector of messages from all of the good senders.  We then use $\netsim$, which can be treated as a function with the powerset monad, to calculate a set of lists of possible messages, which models the effect of the network.  Recall that this set includes the original vector of messages, but also added Byzantine messages, permutations, and dropped messages.  
    We wish to take each possible list of messages in $\mathtt{netmsgs}$ and fold each receiver's combining function and default value over the list of messages.  
    This is accomplished through $\mathsf{map}\ \mathtt{app}\ \mathtt{pairs}$ and $\Pi$ to convert a vector of sets of values to a set of vectors of values.
    Finally, we must place each of these vectors into a record with a field for the receiver role ($\powerset{\lambda x.\{R\mapsto x\}}$).

    It is easy to see that the denotation respects monadic laws, so for instance, $\mathsf{let}$s can be flattened:
    \[
        \denote{\letexp{x_1}{(\letexp{x_2}{p_2}{p_3})}{p_4}} = 
        \denote{\letexp{x_2}{p_2}{\letexp{x_1}{p_3}{p_4}}}
    \]
    
    In fact, every program is equivalent to one in a "let-comm" normal form:
    \[
        \mathit{p} \ \Coloneqq \ \mathtt{ret}\ \{R_i \mapsto e_i \} \ |\  \letexp{x}{\communicate{c}{e_m}{e_d}{e_f}}{\mathit{p}}
    \]
   which we will leverage in the proof.  

    \subsubsection{An Example Program}

    We go through the \texttt{SimpleVote} program in Figure \ref{fig:example-hll} to provide more intuition.  The \texttt{SimpleVote} program is typed by 
    \[
    \tprog{[c:(\Rrole,\Lrole,\bool)]}{[p \hastype \{\Lrole \hastype \bool\}, x \hastype \{\Rrole \hastype \bool\}]}{\{\Lrole, \Rrole\}}{\texttt{SimpleVote}}{\{\Lrole \hastype \texttt{option}\ \bool\}}
    \]
    
    It requires a single Boolean value as input at all good nodes and produces an optional Boolean value at the leader as the output.

    For demonstration, we use a concrete configuration where there is a single leader node and $n_\Rrole=4$ replica nodes, $r_1, r_2, r_3, r_B$, of which $r_B$ is Byzantine and $f_\Rrole=b_\Rrole=1$.
    
    The denotation of \texttt{SimpleVote} is thus a function of type
    \[
    \{p \hastype \{\Lrole \hastype \vectortp{1}{\bool} \}, x \hastype \{\Rrole \hastype \vectortp{3}{\bool}\}\} \rightarrow \powerset{\{\Lrole \hastype \vectortp{1}{(\texttt{option}}\ \bool)\}}
    \]

    We use $\top$ and $\bot$ to represent Boolean true and false.
    Suppose the concrete input is $\{p \mapsto \{\Lrole \mapsto [\top]\} , x \mapsto \{\Rrole \mapsto [\top, \top, \bot] \} \}$, which means the leader's input is $\top$ and the replicas' inputs are $[\top, \top, \bot]$ for $r_1$, $r_2$, and $r_3$ respectively.
    All possible lists of messages that the leader node could receive are given by the $\netsim_\Rrole$ relation.
    Let $\top(r_i)$ represent a value $\top$ sent by node $r_i$.
    The exact set is all permutations of 
    \begin{align*}
        & \{\top(r_1), \top(r_2), \top(r_B)\}, \{\top(r_1), \bot(r_3), \top(r_B)\}, \{\top(r_2), \bot(r_3), \top(r_B)\}, & \\
        & \{\top(r_1), \top(r_2), \bot(r_B)\}, \{\top(r_1), \bot(r_3), \bot(r_B)\}, \{\top(r_2), \bot(r_3), \bot(r_B)\}, & \\
        & \{\top(r_1), \top(r_2), \bot(r_3), \top(r_B)\}, \{\top(r_1), \top(r_2), \bot(r_3), \bot(r_B)\}, \{\top(r_1), \top(r_2), \bot(r_3)\} &
    \end{align*}
    
    Because the input to the leader node is $\top$, $\mathtt{fcnteq}$ counts the number of $\top$ elements in the received messages, and the order of the elements does not influence its output.
    Thus, the set of possible communication results is $\{1, 2, 3\}$.
    $\mathtt{calc\_dec}$ outputs $\Some{\top}$ if there are at least $2=(n - 2 * f)$  votes for $\top$ and $\None$ otherwise.
    So, the set of possible outputs of the leader node is $\{\Some{\top}, \None\}$ given this particular input. This leads to the set of possible final outputs $\{\{\Lrole \mapsto [\Some{\top}]\}, \{\Lrole \mapsto [\None]\} \}$.
    
    \section{Case Studies}
    \label{sec:case_studies}

    In this section, we demonstrate the use of the denotational semantics of \HLL for proving safety properties of two concrete consensus protocols in a Hoare-like style.
    We sketch the steps we followed in our mechanized proof, which is about 2.5k LoC in total, including all the definitions and proofs of the languages and the case studies.
    
    \subsection{Modeling Non-terminating Consensus Protocols}

    In general, the goal of a consensus protocol is to have a collection of nodes agree on a common value, even in the presence of faulty nodes and unpredictable network behaviors.
    The famous FLP impossibility result \cite{fischerImpossibilityDistributedConsensus1985} states that in a fully asynchronous network, there is no non-trivial consensus protocol that is safe, guaranteed to terminate, and can tolerate even a single crash failure. Because of this, consensus protocols are usually not guaranteed to terminate without making stronger assumptions.  

    In our framework, we model protocols that execute a \HLL program in a loop.  
    Since we only consider safety properties, it is sufficient to consider the finite unwindings of the loop.   
    We assume the loop body is a \HLL program that takes in some input and produces an output that contains the input to the next iteration.
    We use the construct $\iter{b}{k}$ to unroll the loop $k$ times and provide fresh channel names for all iterations.

    More formally, a protocol body $b$ is of the form $\nu c_1,\ldots,c_n.\lambda x.p$ where the free channel names in $p$ are bound via the fresh quantifier $\nu$ 
    and $x$ is the input to the protocol.  Given two protocol bodies $b_1$ and $b_2$, their concatenation $b_1 \doubleplus b_2$ is defined as follows:
    \[
      (\nu c_1,\ldots,c_n.\lambda x.p) \doubleplus (\nu c_1',\dots,c_n'.\lambda x'.p') = \nu c_1,\ldots,c_n,c_1',\ldots,c_n'.\lambda x.p[\lambda x'.p'] 
    \]
    where $p[\lambda x'.p']$ is a form of monadic substitution defined by:
    \[ 
      \begin{array}{rcl}
        (\mathtt{ret}\,\{R_i\mapsto e_i\})[\lambda x'.p'] & = & p'[e_i/x'@R_i] \\
        (\letexp{x}{p_1}{p_2})[\lambda x'.p'] & = & \letexp{x}{p_1}{(p_2[\lambda x'.p'])} \\
        (\communicate{c}{e_m}{e_d}{e_f})[\lambda x'.p'] & = & \letexp{x'}{\communicate{c}{e_m}{e_d}{e_f}}{p'} 
      \end{array}
    \]
    
    Note that in the $\mathtt{ret}$ case, our syntax does not support substituting the record expression for the variable $x$, so we must project each role $R$'s expression and substitute that for each occurrence of $x_R$.  With this definition in hand, given a protocol body $b$, we can define $b^k$ as:
    \[
        \begin{array}{rcl}
           b^0 & = & b \\
           b^{k+1} & = & b \doubleplus b^k
        \end{array}
    \]
  
    \subsection{Bosco}

    Bosco~\cite{songBoscoOneStepByzantine2008} stands for \underline{\textbf{B}}yzantine \underline{\textbf{O}}ne-\underline{\textbf{S}}tep \underline{\textbf{CO}}nsensus, a consensus algorithm that only performs a single round of communication per iteration.
    It is a symmetric protocol with only one role, which we denote as $\Rrole$ for replica.
    We assume two global parameters: $n$ is the total number of nodes, and $f$ is the number of Byzantine faults to tolerate.
    For simplicity, we analyze the binary version of the protocol, where the goal is to have all nodes decide on either $\top$ or $\bot$.
    A single iteration of Bosco is defined below:
    \begin{definition}[Binary Bosco in \HLL Syntax]
        \label{alg:Bosco}
        \begin{align*}
            \linen{1} & \mathtt{Bosco}\;(v \hastype \{\Rrole \hastype \bool\}) \hastype \{\Rrole \hastype (\option{\bool})*\bool \} \coloneqq \\
            \linen{2} & \quad \letexp {\mathtt{cnts}} {\communicate{c}{v_\Rrole}{\pure{(0,0)}{\Rrole}}{\pure{\mathtt{fcntb}}{\Rrole}}} {} \\
            \linen{3} & \quad \mathtt{ret}\; (\pure{\mathtt{mkdec}}{\Rrole}\ \mathtt{cnts}_\Rrole) \\
            & \mathsf{where} \\ 
            \linen{4} & \mathtt{fcntb} \coloneqq \lambda\; (\mathtt{cnt}_\top, \mathtt{cnt}_\bot), v.\;\myifthenelse{v}{(\mathtt{cnt}_\top{+}1, \mathtt{cnt}_\bot)}{(\mathtt{cnt}_\top, \mathtt{cnt}_\bot{+}1)} \\
            \linen{5} & \mathtt{mkdec} \coloneqq \lambda\; (\mathtt{cnt}_\top, \mathtt{cnt}_\bot).\; \\
            \linen{6} & \quad \letexp { (\mathtt{newv}, \mathtt{cnt})} {\myifthenelse{\mathtt{cnt}_\top \ge \mathtt{cnt}_\bot}{(\top, \mathtt{cnt}_\top)}{(\bot, \mathtt{cnt}_\bot)}}{} \\
            \linen{7} & \quad \myifthenelse{\mathtt{cnt} * 2 > n + 3 * f}{(\Some{\mathtt{newv}}, \mathtt{newv})}{(\None, \mathtt{newv})} 
        \end{align*}
    \end{definition}

    An iteration of the Bosco protocol requires a single Boolean value at every node as input.
    On line 2, every node broadcasts its input value $v_\Rrole$ to every other node and counts the number of received $\top$ and $\bot$.
    On line 3, each node computes the decision and the input to the next iteration with function $\mathtt{mkdec}$.
    On line 6, in $\mathtt{mkdec}$, a node compares the numbers of $\top$s and $\bot$s it receives and records the value with more votes as the $\mathtt{newv}$ and the number of its occurrences as $\mathtt{cnt}$.
    On line 7, each node produces a decision value $\mathtt{dec}$, which is $\Some{\mathtt{newv}}$ if $\mathtt{cnt}$ is greater than the threshold $\frac{n + 3f}{2}$ or $\None$ if otherwise.
    $\mathtt{newv}$ is to be passed to the next iteration as the input.

    Using our semantic definitions, we prove the following two properties for Bosco:
    \begin{enumerate}
        \item (One Step): When $n > 7f$, if all correct nodes have the same input $B$ in some iteration, then all correct nodes decide and output $\Some{B}$ in that iteration.
        \item (Agreement): When $n > 3f$, 
        if a correct node outputs $\Some{B_1}$ in an iteration and another correct node, not necessarily a different one, outputs $\Some{B_2}$ in a not necessarily different iteration, then $B_1 = B_2$.
    \end{enumerate}

    \begin{definition}[Bosco Configuration]
    We assume a configuration $\mathcal{C}$ with one role $\Rrole$, good nodes $\mathsf{Good}_{\Rrole} = \{ r_1, \ldots, r_{n-b} \}$, and $b$ Byzantine nodes where $b \leq f$. 
    \end{definition}
    
    Proving the one-step property involves a single iteration, and the property is defined formally as:
    \begin{theorem}[Strongly one-step]
        If $n > 7f$, then for all Boolean values $B$, input vectors $\Vec{x}$ such that $\forall r_i, \Vec{x}@r_i = B$, 
        and output values $\{\Rrole \mapsto \Vec{y} \} \in \denote{P} \{x \mapsto \{\Rrole \mapsto \Vec{x}\}\}$, 
        we have $\forall r_i, \Vec{y}@r_i = (\Some{B}, B)$.
    \end{theorem}
    \begin{proof}
        We first give a lemma about $\netsim_\Rrole$.
        For convenience, we define $\#_v(\ell)$ as the number of occurrences of $v$ in the list $\ell$.
        \[
            \#_v(\ell)\coloneqq \foldl{(\mathtt{fcnteq}\ v)}{0}{\ell}\ 
            \text{where}\ 
            \mathtt{fcnteq} \coloneqq \lambda v, c, v'.\; \myifthenelse{v == v'}{c + 1}{c}
        \]
        \begin{lemma}
            \label{lem:fcnteq}
            For any role $\Rrole$ and any $\ell$ such that $|\ell| = n_\Rrole - b_\Rrole$, we have:
            \[\forall \ell' \in \netsim_\Rrole(\ell), v, \#_v(\ell) - f_\Rrole \le \#_v(\ell') \le \#_v(\ell) + b_\Rrole\]
        \end{lemma}
        This lemma can be proven by following the definition of $\netsim_\Rrole$ and $\#_v(\ell)$.
        
        Our main proof follows the structure of the Bosco program's semantics. 
        We start with the precondition $\forall r_i, \Vec{x}@r_i = B$. By definition, line $2$ applied to $\Vec{x}$ reduces to:
        \[
            \Pi(\mathsf{map}\ (\lambda (f, d).\;\powerset{\mathsf{foldl}\ f\ d}\ \netsim_\Rrole(\Vec{x}))\ (\mathtt{const}\ (n - b)\ (\mathtt{fcntb}\ (0,0)))
        \]
        
        So for any output of line $2$, $\mathtt{cnts}$, we have 
        \[
            \forall r_i, \mathtt{cnts}_\Rrole@r_i \in \powerset{\mathsf{foldl}\ \mathtt{fcntb}\ (0,0)}\ \netsim_\Rrole(\Vec{x})
        \]
        
        By the definition of the power set monad, $\powerset{\mathsf{foldl}\ \mathtt{fcntb}\ (0,0)}$ applies the function to every element of the input. Thus,
        \[
            \exists \ell \in \netsim_\Rrole(\Vec{x}), \mathtt{cnts}_\Rrole@r_i = \mathsf{foldl}\ \mathtt{fcntb}\ (0,0)\ \ell
        \]
        
        By definition, $\foldl{\mathtt{fcntb}}{(0, 0)}{\ell} = (\#_\top(\ell), \#_\bot(\ell))$.
        Let $(\mathtt{cnt}_\top, \mathtt{cnt}_\bot) \coloneqq \mathtt{cnts}_\Rrole@r_i$.
        Because of $\forall r_i, \Vec{x}@r_i = B$, $\#_B(\Vec{x})=n-b$ and $\#_{\neg B}(\Vec{x}) = 0$.
        By the lemma above, we know:
        \begin{align*}
            \mathtt{cnt}_B &= \#_B(\ell) \ge \#_B(\Vec{x}) - f = n - b - f \\
            \mathtt{cnt}_{\neg B} &= \#_{\neg B}(\ell) \le \#_{\neg B}(\Vec{x}) + b = b
        \end{align*}

        For line $3$, by the definition of $\mathtt{mkdec}$, because $n > 7f$, we have:
        \begin{align*}
            \mathtt{cnt}_B \ge n - b - f \ge n - 2f > 5f \ge b \ge \mathtt{cnt}_{\neg B}
        \end{align*}
        So, $\mathtt{newv} = B$ and $\mathtt{cnt} = \mathtt{cnt}_B \ge n - 2f > \frac{n + 3f}{2}$, which leads to the final output $(\Some{B}, B)$.
    \end{proof}

    To formalize the agreement property, we first need some auxiliary definitions:
    \begin{align*}
        \mathsf{Step}^k(\Vec{x}) & \coloneqq (\denote{\iter{\mathtt{Bosco}}{k}} \circ \powerset{@\Rrole} \circ \powerset{\mathtt{unzip}}) \{ x \mapsto \{\Rrole \mapsto \Vec{x} \} \} \\
        \mathsf{Decide}_B(\Vec{y}) & \coloneqq \exists r_i. \vec{y}@r_i = \Some{B} \\
        \mathsf{Comply}_B(\Vec{y}) & \coloneqq \forall r_i. \vec{y}@r_i = \Some{B} \vee \vec{y}@r_i = \None
    \end{align*}
    $\mathsf{Step}^k$ runs the protocol for $k + 1$ steps and extracts the results from the record.  
    $\mathsf{Decide}_B$ is a predicate on a vector that holds when there is an element equal to $\Some{B}$, which means a node has decided $B$, and
    $\mathsf{Comply}_B$ is a predicate on a vector that holds when every element is $\Some{B}$ or $\None$, which means each node either decides $B$ or nothing.
    
    We prove a stronger result that implies the agreement property.
    \begin{lemma}[Agreement']
        If $n > 3f$, then 
        \begin{align*}
        & \forall B, \vec{x}, (\vec{y}, \vec{z}) \in \mathsf{Step}^0(\vec{x}), 
                 \mathsf{Decide}_B(\vec{y}) \\ & \quad \Rightarrow 
                 [\mathsf{Comply}_B(\vec{y}) \wedge 
                 \forall k, (\vec{y}_k,\_) \in \mathsf{Step}^{k}(\vec{z}),  
                 \mathsf{Comply}_B(\vec{y}_k)] 
        \end{align*}
    \end{lemma}

    \begin{proof}
        Proofs like these revolve around some \emph{univalent condition}. It is a predicate on the system state.
        All reachable states from a state that satisfies the predicate can only decide on one certain value.
        In our case, thanks to our synchronous functional semantics, the exact univalent condition is exactly the weakest precondition such that the $\mathtt{newv}$ computed on line $6$ can only be a certain value for all the possible sets of messages being received, as follows:
        \[
        UC'_B(\vec{x}) \coloneqq \forall \vec{x'} \in \netsim_\Rrole(\vec{x}), \mathsf{snd}(\mathtt{mkdec}(\foldl{\mathtt{fcntb}}{(0,0)}{\vec{x'}}))=B
        \]
        
        Due to the asymmetrical nature of the $\le$ comparison, this is equivalent to:
        \[
        UC'_\top(\vec{x}) = \#_\top(\vec{x}) \ge \frac{n+f}{2} \quad
        UC'_\bot(\vec{x}) = \#_\bot(\vec{x}) > \frac{n+f}{2}
        \]

        The decision procedure for Bosco does not utilize this asymmetry~\cite{songBoscoOneStepByzantine2008}.
        An optimization is possible by using $\ge$ instead of $>$ in the condition on line $7$ only when $\mathtt{newv}$ is $\top$.
        Here, for simplicity, we use a slightly stronger symmetric predicate that is still sufficient to prove our result:
        \[UC_B(\vec{x}) \coloneqq \#_B(\vec{x}) > \frac{n+f}{2}\]
        
        And we prove the agreement property by proving the following:
        \begin{align}
            & \forall \vec{x}, (\vec{y},\vec{z}) \in \mathsf{Step}^0(\vec{x}), \mathsf{Decide}_B(\vec{y}) \Rightarrow UC_B(\vec{x}) &\\
            & \forall \vec{x}, k, (\vec{y}, \vec{z}) \in \mathsf{Step}^k(\vec{x}), UC_B(\vec{x}) \Rightarrow [\mathsf{Comply}_B(\vec{y}) \wedge UC_B(\vec{z})] &
        \end{align}
        
        We prove $(1)$ by walking through the program backward with $\mathsf{Decide}_B(\vec{y})$.
        Because our program is nondeterministic, walking through the program backward requires us to infer a predicate that covers all possible input values that can lead to any output values satisfying the postcondition.

        \begin{itemize}
            \item On line $3$, we know that for some $r_i$, $\mathtt{dec} = \Some{B}$.
            \item By the definition of $\mathtt{mkdec}$, it is necessary to have $\mathtt{cnt}_B > \frac{n+3f}{2}$.
            \item On line $2$, by Lemma \ref{lem:fcnteq}, $\#_B(\vec{x}) \ge \mathtt{cnt}_B - b > \frac{n+f}{2}$, which is our goal, $UC_B$.
        \end{itemize}

        For $(2)$, we do an induction on $k$.
        The only case that needs to be proven is $k = 0$.
        We walk through the program forward to infer a predicate that covers all possible outputs.
        \begin{itemize}
            \item On line $2$, by Lemma \ref{lem:fcnteq} and $UC_B(\vec{x})$, for any $r_i$ and any possible network, we have $\mathtt{cnt}_B > \frac{n+f}{2} - f \wedge \mathtt{cnt}_{\neg B} \le \#_{\neg B}(\vec{x}) + b < n + f - \frac{n+f}{2} = \frac{n + f}{2}$. So $\mathtt{cnt}_B > \mathtt{cnt}_{\neg B}$.
            \item By the definition of $\mathtt{mkdec}$, $\mathtt{newv} = B$ and $\mathtt{dec} = \Some{\mathtt{newv}} = \Some{B} \vee \mathtt{dec} = \None$.
            \item On line $3$, we have $\forall r_i,\ (\vec{y}@r_i = \Some{B} \vee \vec{y}@r_i = \None) \wedge (\vec{z}@r_i = B)$.
        \end{itemize}
        $\forall r_i,\ \vec{y}@r_i\Some{B} \vee \vec{y}@r_i = \None$ implies $\mathsf{Comply}_B(\vec{y})$.
        $\forall r_i,\ \vec{z}@r_i = B$ implies $\#_B(\vec{z}) = n - b \ge n - f$.
        Because $n > 3f$, we have $\#_B(\vec{z}) > \frac{n+f}{2}$, which is $UC_B$.
    \end{proof}

    \subsection{Sequential Paxos}

    Sequential Paxos, or SeqPaxos for short, is a variant of the Synod consensus protocol used by Paxos~\cite{lamportParttimeParliament1998}.
    Similar to the original single-decree Paxos protocol, it is a crash fault-tolerant protocol with a single leader node and an arbitrary number of so-called acceptor nodes, which we will call \emph{replicas}.
    
    Let $n$ be the total number of replicas and $f$ be the number of replica crashes to tolerate.
    The single leader may also crash.
    We assume the domain of the consensus value is of type $\mathds{V}$, for which a decidable equality exists. 
    Additionally, we assume that there is a default value for the leader to propose in each iteration, denoted as a global variable $\mathtt{default} \hastype \nat \funarrow \mathds{V}$.
    A single iteration of SeqPaxos defined in \HLL is listed below:
    \begin{definition}{SeqPaxos}
        \begin{align*}
            \linen{1} & \mathtt{SeqPaxos}\; (x \hastype \{\Lrole \hastype \nat, \Rrole \hastype (\option{\mathds{V}}) \tighten{*} \nat \}) \\ 
                      & \qquad \hastype \{\Lrole \hastype (\option{\mathds{V}}) \tighten{*} \nat, \Rrole \hastype (\option{\mathds{V}}) \tighten{*} \nat \} \}  \coloneqq \\  
            \linen{2} & \quad \letexp{\mathtt{maxv}}{\communicate{c_1}{x_\Rrole}{\pure{(\None, 0)}{\Lrole}}{\pure{\mathtt{fmaxr}}{\Lrole}}}{} \\
            \linen{3} & \quad \letexp{p}{\mathtt{ret}\; \{\Lrole \mapsto \pure{\mathtt{pickp}}{\Lrole}\ \mathtt{maxv}_\Lrole\ (\pure{\mathtt{default}}{\Lrole}\ (x_\Lrole))\}}{} \\
            \linen{4} & \quad \letexp{y}{\communicate{c_2}{(\pure{\mathtt{pair}}{\Lrole}\ p_\Lrole\ x_\Lrole)}{x_\Rrole}{\pure{\mathtt{update}}{\Rrole}}}{}\\
            \linen{5} & \quad \letexp{\mathtt{cnt}}{\communicate{c_3}{y_\Rrole}{\pure{0}{\Lrole}}{(\pure{\mathtt{fcnteq}}{\Lrole}\ x_\Lrole)}}{} \\
            \linen{6} & \quad \mathtt{ret}\; \{\Lrole \mapsto \pure{\mathtt{pair}}{\Lrole}\ (\pure{\mathtt{mkdec}}{\Lrole}\ \mathtt{cnt}_\Lrole\ p_\Lrole) \ (\pure{\mathtt{add}}{\Lrole}\ x_\Lrole\ \pure{1}{\Lrole}), \Rrole \mapsto y_\Rrole\}\\
            & \mathsf{where} \\
            \linen{7} & \quad \mathtt{fmaxr} \coloneqq \lambda\; (v, r), (v', r').\;\myifthenelse{r < r'}{(v', r')}{(v, r)}\\
            \linen{8} & \quad \mathtt{pickp} \coloneqq \lambda\; (ov, \_), d. \; \mathsf{match}\ ov\ \mathsf{with}\ |\;\Some{v} \Rightarrow v\ |\;\None \Rightarrow d\ \mathsf{end} \\
            \linen{9} & \quad \mathtt{update} \coloneqq \lambda\; \_, (v, r).\; (\Some{v}, r)\\
            \linen{10} & \quad \mathtt{fcnteq} \coloneqq \lambda\; r, c, (\_, r').\; \myifthenelse{r == r'}{c + 1}{c}\\
            \linen{11} & \quad \mathtt{mkdec} \coloneqq \lambda\; c, p.\; \myifthenelse{c > f}{\Some{p}}{\None}
        \end{align*}
    \end{definition}
    
    In the SeqPaxos protocol, the leader takes a single natural number, which is the current round number.
    The replicas take a pair of an optional $\mathds{V}$ value and a natural number.
    The number is the latest round number ever received, and the value is the proposal of that round.
    To initiate the protocol, the leader should be given the input $1$ and the replicas the dummy values $\None$ and $0$.

    SeqPaxos performs three rounds of communication in each iteration.
    In the first round of communication on line $2$, all the replicas send their local value and round number to the leader.
    The leader then finds the largest round number and the proposal associated with the round number.
    On line $3$, the leader computes the proposal of the iteration.
    If the proposal associated with the largest round number received is $\Some{v}$, the proposal is $v$.
    Otherwise, the proposal is set to the default value of that iteration.
    On line $4$, the leader broadcasts the proposal with its round number to all the replicas.
    For each replica, if it receives the message from the leader, it updates its local value and round number to the new local value and round number; otherwise, it keeps its old local value and round number.
    On line $5$, all the replicas send their local value and round number to the leader again.
    This time, the leader counts how many replicas have updated their local value and round number to the latest proposal and round number pair.
    On line $6$, if the number of up-to-date replicas is greater than $f$, the leader decides their proposal as the decision of the iteration; otherwise, no decision is made.
    The leader returns a pair of the decision and an incremented round number while each replica returns its latest value.

    \begin{definition}[SeqPaxos Configurations]
        SeqPaxos requires a configuration $\mathcal{C}$ to have two roles, $\Lrole$ and $\Rrole$.
        $\Lrole$ is the leader role, $\mathsf{Good}_\Lrole = \{l\}$, $b_\Lrole = 0$. 
        We set $f_\Lrole = 1$ to model that the leader may crash.
        $\Rrole$ is the replica role, $\mathsf{Good}_\Rrole = \{r_1, \dots, r_n\}$.
        We also set $b_\Rrole = 0$ and $f_\Rrole = f$ to model the assumption that up to $f$ replicas may crash.
    \end{definition}
    While we do not directly model the crash behavior, we claim the above definition covers all possible cases for SeqPaxos.
    This is because we are only reasoning about finite iterations, and there is no difference between a node that has crashed at some point and a node that has certain messages sent from and to it dropped by $\netsim$.

    The special initial input to SeqPaxos is defined as $\mathtt{init} = \{\Lrole \mapsto [1]; \Rrole \mapsto \vec{x} \}$ where $\vec{x}$ is $n$ copies of $(\None, 0)$.

    We formally state and prove the agreement property for SeqPaxos.
    \begin{theorem}[Agreement]
        If $n > 2f$, 
        \begin{align*}
            & \forall D, i, \{\Lrole \mapsto [(d_i, r_i)] ;\Rrole \mapsto \vec{x}_i\} \in \denote{\iter{\mathtt{SeqPaxos}}{i}}(\{x \mapsto \mathtt{init}\}), \ \  d_i = \Some{D} \Rightarrow & \\
            & \quad \quad \forall j > i, \{\Lrole \mapsto [(d_j, \_)] ;\Rrole \mapsto \_\} \in \denote{\iter{\mathtt{SeqPaxos}}{j - i - 1}}(\{x \mapsto \{\Lrole \mapsto [r_i]; \Rrole \mapsto \vec{x}_i\}\}), & \\
            & \quad \quad d_j = \Some{D} \vee d_j = \None &
        \end{align*}
    \end{theorem}
    
    \begin{proof}
        We have the following lemma about the inputs to any iteration $i$, which can be proven by induction on $i$ and stepping through the program.
        \begin{lemma}
            \label{lem:basic_properties}
            $\forall i, \{\Lrole \mapsto [(\_, r)] ;\Rrole \mapsto \vec{x}\} \in \denote{\iter{\mathtt{SeqPaxos}}{i}}(\{x \mapsto \mathtt{init}\})$, implies:
            \begin{itemize}
                \item $r=i + 2\ \wedge$
                \item $\forall u, \mathsf{snd}(\vec{x}@u) < i + 2\ \wedge$
                \item $\forall u, 0 < \mathsf{snd}(\vec{x}@u) \Rightarrow \exists v, \mathsf{fst}(\vec{x}@u) = \Some{v}\ \wedge $
                \item $\forall u, w, \mathsf{snd}(\vec{x}@u) = \mathsf{snd}(\vec{x}@w) \Rightarrow \vec{x}@u = \vec{x}@w$.
            \end{itemize}
        \end{lemma}
        We again define the univalent condition $UC_D$ for SeqPaxos as the weakest precondition\footnote{We treat $\mathtt{default}$ as some value that cannot be used outside of the particular round at runtime.} of line $3$ always resulting in value $D$.
        \[ 
            UC_D(\vec{x}) \coloneqq \forall \ell \in \netsim_\Rrole(\vec{x}), \mathsf{fst}(\foldl{\mathtt{fmaxr}}{(\None, 0)}{\ell}) = \Some D
        \]

        For the agreement property, we prove the following:
        \begin{align*}
            & \forall D, i, \{\Lrole \mapsto [(d_i, \_)] ;\Rrole \mapsto \vec{x}_i\} \in \denote{\iter{\mathtt{SeqPaxos}}{i}}(\{x \mapsto \mathtt{init}\}), \\
            & \quad d_i = \Some{D} \Rightarrow UC_D(\vec{x}_i) & (1) \\
            & \forall D, i, \{\Lrole \mapsto \_ ;\Rrole \mapsto \vec{x}_i\} \in \denote{\iter{\mathtt{SeqPaxos}}{i}}(\{x \mapsto \mathtt{init}\}), UC_D(\vec{x}_i) \Rightarrow & \\ 
            & \quad \forall j > i, \{\Lrole \mapsto [(d_j, \_)] ;\Rrole \mapsto \vec{x}_j\} \in \denote{\iter{\mathtt{SeqPaxos}}{j - i - 1}}(\{x \mapsto \{\Lrole \mapsto [r_i]; \Rrole \mapsto \vec{x}_i\}\}), & \\ 
            & \quad \quad (d_j = \Some{D} \vee d_j = \None) \wedge UC_D(\vec{x}_j)\ & (2)
        \end{align*}

        It is worth noting that our UC predicate is weaker than the strongest postcondition of lines $5$ and $6$ deciding $D$, which means the protocol is actually safe to decide in more cases. An optimization for the original Paxos protocol that utilizes this is described in \cite{goldweberBriefAnnouncementSignificance2020}.
        Here, we prove the property with this precise condition, allowing optimizations that weaken the decision condition to be verified compositionally and reuse the proof of (2).

        To prove $(1)$, we work backward to find all possible intermediate cases that may lead to $d_i = \Some{D}$ and then prove they entail $UC_D(\vec{x}_i)$.
        \begin{itemize}
            \item By Lemma \ref{lem:basic_properties}, the round number input for the round that produces $\vec{x}_i$ is $i + 1$.
            \item By the definition of $\mathtt{mkdec}$ on line $11$, we know on line $6$, for the single leader node $l$, $\mathtt{cnt} > f \wedge p = D$.
            \item On line $5$, we have \[ \exists \ell \in \netsim_\Rrole(\vec{x}_i), \foldl{(\mathtt{fcnteq}\ (i + 1))}{0}{\ell} > f\]
            \item By line $4$, Lemma \ref{lem:fcnteq}, and $b_\Rrole = 0$, $\#_{(\Some{D}, i+1)}(\vec{x}_i) > f - b_\Rrole = f$. 
            \item Unfolding the definition of $UC_D$, we need to prove:
            \[ \forall \ell \in \netsim_\Rrole(\vec{x}_i), \mathsf{fst}(\foldl{\mathtt{fmaxr}}{(\None, 0)}{\ell}) = \Some D \]
            By Lemma \ref{lem:fcnteq} again, $\#_{(\Some{D}, i+1)}(\ell) > f - f = 0$.
            So there exists $w$, $w \in \ell = (\Some{D}, i+1)$.
            By Lemma \ref{lem:basic_properties},
            \[\forall u, \mathsf{snd}(\vec{x}_i@u) \le i + 1 \wedge \mathsf{snd}(\vec{x}_i@u) = i + 1 \Rightarrow \vec{x}_i@u = \vec{x}_i@w = (\Some{D}, i + 1)\]
            which means $i + 1$ is the largest round number, and every replica with round number $i + 1$ has value $(\Some{D}, i + 1)$.
            So by the definition of $\mathtt{fmaxr}$,
            \[\foldl{\mathtt{fmaxr}}{(\None, 0)}{\ell} = (\Some{D}, i + 1)\]
        \end{itemize}

        To prove $(2)$, we perform induction on $j$.
        It suffices to prove for a single iteration, and we do a forward pass through the program.
        \begin{itemize}
            \item By Lemma \ref{lem:basic_properties}, the round number output for the round that produces $\vec{x}_i$ is $i + 2$.
            \item For lines $2$ and $3$, by the definition of $UC_D$, we have $p = D$. This implies $\mathtt{dec} = \Some{D} \vee \mathtt{dec} = \None$ on line $6$.
            \item On line $4$, our goal is to prove $UC_D$ holds on the new values of the replicas, $y@\Rrole$.
            By the definition of $\mathtt{update}$ and Lemma \ref{lem:basic_properties}, 
            \[\forall u, y@\Rrole@u = (\Some{D}, i + 2) \vee y@\Rrole@u = \vec{x}_i@u \]
            So each node nondeterministically picks between switching to the new value $(\Some{D}, i + 2)$ and keeping the old value $\vec{x}_i@u$.
            This property is preserved by the $\netsim$ relation, more formally:
            \[
                \forall \ell \in \netsim_\Rrole(y@\Rrole), \exists \ell' \in \netsim_\Rrole(\vec{x}_i), \forall u, \ell@u = (\Some{D}, i) \vee \ell@u = \ell'@u
            \]
            
            For any $\ell \in \netsim_\Rrole(y@\Rrole)$, there are two cases.
            \item If $(\Some{D}, i + 2) \in \ell$, then $\foldl{\mathtt{fmaxr}}{(\None, 0)}{\ell} = (\Some D, i + 2)$.
            \item Otherwise, all nodes whose messages were received did not update their local values. Thus, $\ell = \ell'$ and $\ell \in \netsim_\Rrole(\vec{x}_i)$.
            Because of $UC_D(\vec{x}_i)$, $\mathsf{fst}(\foldl{\mathtt{fmaxr}}{(\None, 0)}{\ell}) = \Some D$.
            Thus, $UC_D(\vec{x}_{i+1}=y@\Rrole)$.
        \end{itemize}
    \end{proof}

    \section{Low-Level Semantics}
    \label{sec:operational}

    In this section, we define \LLL, whose semantics is a labeled transition system that serves as a low-level operational semantics for general distributed systems, and that serves as a target language for compiling \HLL programs.  
    \LLL provides a formal basis for proving the adequacy of the denotational reasoning we propose for \HLL programs.
    
    We first clarify the network and adversary assumptions used \LLL's semantics is modeled after
    and then construct the semantics by composing multiple instances of local labeled transition systems, \LLL nodes and \LLL channels.
    An \LLL node models the behaviors of an individual non-Byzantine node in the system, while 
    an \LLL channel models the influence of the network and Byzantine nodes in a single logical round of communication.   

    \subsection{Network and Adversary Assumptions}
    \label{sec:network}

    We assume any network message contains two parts of information: some message payload $x$, which is a well-typed value, and an auxiliary header, which includes information such as an identifier for the channel this message is for, the sender's identity, the receiver's identity, \emph{etc.}
    The header information is invisible to the distributed program and therefore not modeled in the semantics; it is used only by the runtime system to provide specific guarantees.

    We assume the network does not duplicate messages, which can be enforced by suppressing the duplicates at the receiving end. 
    We also require the messages to be \emph{authenticated}: if a message with the payload $x$ is received by some non-Byzantine receiver $p$, the runtime of $p$ can check that
    the identity of the sender is some node $q$ in the system using the information in the header.
    Furthermore, if $q$ is non-Byzantine, $q$ must have sent the message with payload $x$.
    Cryptographic signatures are a common mechanism used to satisfy these assumptions, but we do not model signatures directly.
    
    For the delivery of messages, we assume the network is asynchronous and reliable.
    Asynchronous means that there is no bound on the delay between the sender node sending a message and the receiver node receiving and processing this message.
    Reliable means that a message sent from a correct sender to a correct receiver will eventually arrive, which is necessary for the system to make progress.
    This eventual arrival can be achieved by resending the messages.

    We assume a powerful adversary that overestimates the capabilities of a realistic one.
    In particular, our adversary has complete control over the set of faulty nodes, has sufficient computational resources, is capable of sending out arbitrary messages, and controls the message delivery schedule.
    
    However, we do not explicitly model the adversary doing any of the following:
    \begin{itemize}
        \item Attacking the liveness of the network, such as in denial-of-service attacks.
        \item Modifying messages sent by non-Byzantine nodes as in man-in-the-middle attacks.
        Our authenticated network assumption excludes the cases where the adversary forges messages that appear to have been sent by a correct node.
        \item Sending ill-formed messages, such as incomplete auxiliary information or an ill-typed value.
        Although a common source of vulnerabilities in reality, this kind of messages can be filtered by a correct parser and dropped.
        In our case, we do not model parsing nor verify its correctness.
        \item Sending multiple messages with the same header but different payloads.  We assume the receiver simply drops all but one message from a given sender.
    \end{itemize}
    
    Under our assumptions, the Byzantine adversary can be modeled as a stateless entity because any series of events that can happen during the execution does not change its capabilities.
    Furthermore, because the only ways the adversary can affect the execution of correct nodes are by sending well-formed messages that contain arbitrary contents and manipulating the delivery schedule, we model the adversary's influence through a special action of the network that creates well-formed messages out of thin air and consider all possible schedules of the network.


    \subsection{\LLL Node Semantics}
    \label{sec:thread}

    Recall that a labeled transition system $\langle S, \Lambda, \rightarrow \rangle$ includes
    a set of states $S$, a set of labels $\Lambda$, and a transition relation $\rightarrow$ over $S \times \Lambda \times S$ that relates a starting state, a label, and a next state.
    A behavior or \emph{trace} of a labeled transition system from a specific initial state is a sequence of labels and states where each element is permitted by $\rightarrow$ from the previous state, given the label.  We write $s \xrightarrow{\ell} t$ for when $(s, \ell, t) \in \rightarrow$.  If $L$ is a sequence of labels, we write 
    $s \xrightarrow{L} t$ to mean if $L$ is empty, then $t = s$, and if $L = \ell::L'$, then there exists a $t'$ such that $s \xrightarrow{\ell} t'$ and $t' \xrightarrow{L'} t$.
    
    The core of an \LLL node is a program represented by the following definition in Rocq, which we will use to represent states in our transition system for nodes.
    \begin{definition}[\LLL Monad for \LLL Node Syntax]
        \begin{equation*}
            \renewcommand{\arraystretch}{1.2}
            \begin{array}{l}
                \mathtt{Inductive}\ \mathsf{\LLL} (T\hastype\mathtt{Type}) \hastype \mathtt{Type} \coloneqq \\
                \quad \mid \mathsf{return}\hastype T \funarrow \mathsf{\LLL}\ T \\
                \quad \mid \mathsf{sendThen} \hastype {\forall}c\hastype\mathsf{Channel}, \mathsf{msg\_t}(c) \funarrow \mathsf{\LLL}\ T \funarrow \mathsf{\LLL}\ T \\
                \quad \mid \mathsf{rcvThen} \hastype \forall c\hastype\mathsf{Channel}, (\mathsf{list}(\mathsf{msg\_t}\ c) \funarrow 
                    \mathsf{\LLL}\ T) \funarrow \mathsf{\LLL}\ T. \\
                    \\
                \mathtt{Fixpoint}\ \mathsf{bind}\{T\ U\}(e\hastype \mathsf{\LLL}\ T) (f\hastype T \funarrow \mathsf{\LLL}\ U) \coloneqq \\
                \quad \mathtt{match}\ e\ \mathtt{with} \\
                \quad \mid \mathsf{return}\ v \Rightarrow f v \\
                \quad \mid \mathsf{sendThen}\ c\ m\ k \Rightarrow\  \mathsf{sendThen}\ c\ m\ (\mathsf{bind}\ k\ f) \\
                \quad \mid \mathsf{rcvThen}\ c\ g\ \Rightarrow\ \mathsf{rcvThen}\ c\ (\lambda\ m.\ \mathsf{bind}\ (g\ m)\ f) \\
                \quad \mathtt{end}. \\ 
                \\
                \mathtt{Definition}\ \mathsf{send}\ c\ m\ \coloneqq\ \mathsf{sendThen}\ c\ m\ (\mathsf{return}\ \mathsf{tt}). \\
                \mathtt{Definition}\ \mathsf{receive}\ c\ (d\hastype\ T)\ (f\hastype\ T \funarrow \mathsf{msg\_t}\ c\funarrow\ T) 
                \coloneqq \mathsf{rcvThen}\ c\ (\lambda m.\ \mathsf{return}\ (\mathsf{foldl}\ f\ d\ m)). \\
                
            \end{array}
        \end{equation*}
    \end{definition}

    $\mathsf{\LLL}\ T$ is a recursively defined type, which describes three kinds of computations: returning a value of type $T$; sending a message to channel $c$, and then continuing with another computation; receiving from a channel $c$ a list of messages, which are fed to a function to produce a continuation. 

    Some auxiliary operations will be needed to build \LLL node programs below.  
    The operation $\mathsf{bind}\ t\ f$ ``appends'' $f$ onto the leaves of \LLL node program $t$.  More properly, the leaves of $t$ must be of the form $\mathsf{return}\ e$, and the $\mathsf{bind}$ replaces these leaves with $f\ e$. The definitions for $\mathsf{send}$ and 
    $\mathsf{receive}$ build corresponding \LLL node program trees with trivial continuations that immediately return.  Note that $\mathsf{bind}$ and $\mathsf{return}$ form a monad.
        
    To give meaning to \LLL node programs, we formulate a labeled transition system $\langle S_t, \Lambda_t, \rightarrow_t \rangle$:

    \begin{definition}[\LLL Node Semantics]
        \label{def:single_node}
        The state, $S_t$ is an \LLL node program of some return type $T$.
        
        $\Lambda_t$ has two kinds of labels: 
        \begin{align*}
            l_t & \Coloneqq \mathsf{send}\ c\ v \mid \mathsf{receive}\ c\ \Vec{v}
        \end{align*}
        
        The transition relation is defined by:  
        \begin{mathpar}
            \infer*[left=(Send)]
            {}
            {\mathsf{sendThen}\ c\ v\ k \xrightarrow{\mathsf{send}\ c\ v} k}
            
            \infer*[left=(Receive)]
            {} 
            {\mathsf{rcvThen}\ c\ f \xrightarrow{\mathsf{receive}\ c\ \Vec{v}} f(\vec{v})}
        \end{mathpar}
        
    \end{definition}

    To obtain the initial state for a given node, we now define a translation from \HLL programs and roles to \LLL node programs, which is a form of \emph{endpoint projection} because \HLL is in choreography style.  
    
    \begin{definition}[Compiling \HLL to \LLL Node Programs]
        \label{def:HLL_to_LLL}
        Given a record type $\tau = \{R_1\hastype\tau_1,\ldots,R_n\hastype\tau_n\}$ we define $\compile{\tau}{R_i} = \tau_i$.  
        When $R$ is not a field of $\tau$, we take $\compile{\tau}{\Rrole}$ to be $\mathsf{unit}$.  
        Given a variable context $\Gamma = [x_1\hastype\tau_1,\ldots,x_n\hastype\tau_n]$ we define $\compile{\Gamma}{R}$ to be 
        $[x_1\hastype\compile{\tau_1}{R},\ldots,x_n\hastype\compile{\tau_n}{R}]$.  Then, 
        given a \HLL program $p$ such that $\tprog{\Delta}{\Gamma}{\mathcal{R}}{p}{\tau}$ and a role $\Rrole \in \mathcal{R}$ and node $i$ in role $R$,
        the compilation of $p$ at role $\Rrole$ and node $i$
        denoted by $\compile{p}{\Rrole,i}$ yields a term $t$ such that $\compile{\Gamma}{\Rrole} \vdash t : \mathsf{\LLL}\compile{\tau}{\Rrole}$ and is defined by:
        \[       
        \begin{array}{rcl}
            \compile{\mathtt{ret}\,\{R_j\mapsto e_j\}}{R,i} & = & \begin{cases}
                                                                    \mathsf{return}\ \compile{e_j}{R,i} \qquad &  R = R_j \\
                                                                    \mathsf{return}\ \mathsf{tt} & \forall j, R \neq R_j
                                                                   \end{cases} \\ \\
            \compile{\letexp{x}{p_1}{p_2}}{R,i} & = & \compile{p_1}{R,i} \bind \lambda x.\compile{p_2}{R,i} \\ \\            
            \compile{\communicate{c}{e_m}{e_d}{e_f}}{R,i} & = & \begin{cases}
                                                            \mathsf{let}\ \_ \leftarrow \mathsf{send}\ c\ \compile{e_m}{R,i}\ \mathsf{in} \\ 
                                                            \quad \mathsf{receive}\ c\ \compile{e_d}{R,i}\ \compile{e_f}{R,i}
                                                                & \doublebackslash R\ \text{sending and receiving} \\
                                                            \mathsf{send}\ c\ \compile{e_m}{R,i}
                                                                & \doublebackslash R\ \text{sending but not receiving} \\
                                                            \mathsf{receive}\ c\ \compile{e_d}{R,i}\ \compile{e_f}{R,i}
                                                                & \doublebackslash R\ \text{receiving but not sending} \\
                                                             \mathsf{return}\ \mathsf{tt}
                                                                & \doublebackslash R\ \text{not sending nor receiving} \\
                                                          \end{cases}
        \end{array}
        \]
    \end{definition}
        The definition assumes a translation for expressions $\compile{e}{R,i}$ which simply removes the $R$ from variables and terms:
        \[
        \begin{array}{rcl}
            \compile{x_R}{R,i} & = & x \\
            \compile{\pure{t}{R}}{R,i} & = & t \\
            \compile{[t_1,\ldots,t_n]}{R,i} & = & t_i \\
            \compile{e_1\ e_2}{R,i} & = & \compile{e_1}{R,i}\ \compile{e_2}{R,i}
        \end{array}
        \]
        
    Finally, assuming a top-level program $p$ is closed, then given a role $R$ and node $i$ for that role, we can form a closed \LLL node program for the initial state of node $i$ by 
    computing $\compile{p}{R,i}$.

    As an example, the \LLL node programs for the two roles in the $\mathtt{SimpleVote}$ program in Figure~\ref{fig:example-hll} are shown below, using the notation $\mathsf{let}\ x \coloneqq e_1\ \mathsf{in}\ e_2$ to represent 
    $\mathsf{bind}\ e_1\ (\lambda\ x.e_2)$:
    \begin{center}
    \begin{tabular}{l|l}
        \Pseudocomment{\LLL node program for the leader:} &  \Pseudocomment{\LLL node program for the (correct) replicas:} \\
        $\lambda p.\;\textsf{let}\ \mathtt{cnt}\; \coloneqq \; \mathsf{receive}\ c\ 0\ (\mathtt{fcnteq}\ p)\ \mathsf{in}$
        & $\lambda x.\; \mathsf{send}\ c\ x$
        \\
        $\quad \mathsf{return}\ (\mathtt{calc\_dec}\ \mathtt{cnt}\ x)$
        & 
    \end{tabular}
    \end{center}
    
    The execution of a single \LLL node program requires an environment that includes other nodes and some model that captures the intuitive behavior for communication mentioned above.  To achieve this, our next step is to build a transition system that models communication channels and then to compose \LLL node and channel transition systems into a global transition system.  
    
    \subsection{Single-use Channel Semantics}
    \label{sec:channel}

    A \emph{channel} is an abstraction used to describe the network. Instead of being modeled as a single entity, the network is logically divided into separate communication channels, each serving only a fixed set of senders and receivers.
    Channels simplify reasoning by separating messages sent for different logical purposes.
    We model network behaviors using multiple single-use channels instead of the more common multi-use channels.
    Each single-use channel models a single logical round of communication, and different channels correspond to different logical communication steps.
    We associate channels with statically known information and give them a dynamic operational semantics.
    The list of parameters that defines a channel is as follows:

    \begin{definition}[\LLL Channel Parameters]
        Assume we are given a role set $\mathcal{R}$, a configuration $\mathcal{C}$ for $\mathcal{R}$, and a channel context $\Delta$. 
        Recall that each channel appears at most once in $\Delta$
        and is associated with a message type $\tau$, as well as a sender role and a receiver role, which can be the same role.
        We use $\mathsf{msg\_t}_{\Delta}(c)$ to represent the type $\tau$ associated with $c$ in $\Delta$.  

        Suppose $c$ is a channel associated with sender role $S$ and receiver role $\Rrole$.  Then 
        $\mathsf{sender}(c)$ denotes the set of non-Byzantine nodes in configuration $\mathcal{C}$ associated with role $S$;  
        $\mathsf{Byz\_sender}(c)$ denotes the set of Byzantine sender nodes in $\mathcal{C}$ associated with role $S$; and
        $\mathsf{receiver}(c)$ denotes the set of non-Byzantine receiver nodes in $\mathcal{C}$ associated with role $\Rrole$.
    \end{definition}

    We also define $\netsim_{\mathsf{\LLL}}$, which models the effect of asynchrony of the network in the same language of $\netsim_\Rrole$.
    \begin{equation*}
        \netsim_{\mathsf{\LLL}} \hastype \mathsf{list}\;\mathsf{msg\_t}(c) \funarrow \powerset{\mathsf{list}\;\mathsf{msg\_t}(c)} \coloneqq \lambda l.\ (\mathsf{perm}\ l) \bind (\lambda l.\ \mathsf{trunc}\, l\ (n_S-f_S))
    \end{equation*}
    
    $\netsim_{\mathsf{\LLL}}$ is a function that maps a list of messages sent to a receiver node to the set of possible lists of messages the receiver could possibly receive.
    Under our network assumptions, it only performs two operations: permutation, which models message reordering in transit, and truncation of the list to a prefix of at least the lower bound on the number of correct sender nodes, i.e., the nodes that are neither Byzantine nor crash.
    The effects of Byzantine nodes are separately captured by the transition relation.  

    Given a channel $c$ with the parameters defined above, we define the semantics of \LLL channel as the labeled transition system $\langle S_c, \Lambda_c, \rightarrow_c \rangle$ below.

    \begin{definition}[\LLL Channel Semantics]
        \label{def:channel}
        A channel state $s \in S_c$ is a tuple of four components:
        \begin{enumerate}
            \item $F_s \subseteq \mathsf{sender}(c)$, the set of non-Byzantine senders who have performed their send action.
            \item $F_r \subseteq \mathsf{receiver}(c)$, the set of non-Byzantine receivers who have done their receive action.
            \item $F_b \hastype \mathsf{receiver}(c) \funarrow \powerset{\mathsf{Byz\_sender}(c)}$, for each non-Byzantine receiver node, a set of Byzantine sender nodes who have sent the receiver a message.
            \item $M \hastype \mathsf{receiver}(c) \funarrow \mathsf{list}\;\mathsf{msg\_t}(c)$, for each receiver a list of message payloads that have been sent to it.
        \end{enumerate}

        $\Lambda_c$ has three kinds of labels:
        \begin{equation*}
            l_c \Coloneqq \mathsf{send}\ s\ v
                \mid \mathsf{byz\_send}\ sb\ r\ v
                \mid \mathsf{receive}\ r\ \Vec{v} 
        \end{equation*}
        where $s \in \mathsf{sender}(c)$, $sb \in \mathsf{Byz\_sender}(c)$, $r \in \mathsf{receiver}(c)$,
        $v \hastype \mathsf{msg\_t}(c)$, and $\Vec{v} \hastype \mathsf{list}\;\mathsf{msg\_t}(c)$.
        
        $\mathsf{send}\ s\ v$ models a non-Byzantine sender node $s$ broadcasting a message $v$ to all the receiver nodes. $\mathsf{receive}\ r\ \Vec{v}$ models a non-Byzantine receiver node $r$ receiving a well-formed list of messages $\Vec{v}$. $\mathsf{byz\_send}\ sb\ r\ v$ models a Byzantine sender node $sb$ sends a single message of content $v$ to a single non-Byzantine receiver node $r$.        
        
        $\rightarrow_c$ has one rule for each kind of label.
        \begin{mathpar}
            \infer*[left=send]
            {s \notin F_s \wedge F_s' = F_s \cup \{s\} \wedge \forall r, M'(r) = M(r) \doubleplus [v] }
            {\langle F_s, F_r, F_b, M \rangle \xrightarrow{\mathsf{send}\ s\ v} \langle F_s', F_r, F_b, M' \rangle} \\
            \infer*[left=byz\_send]
            {\forall r' \ne r, F'_b(r') = F_b(r') \wedge M'(r') = M(r') \\ sb \notin F_b(r) \wedge F'_b(r) = F_b(r) \cup \{sb\} \wedge M'(r) = M(r) \doubleplus [v] }
            {\langle F_s, F_r, F_b, M \rangle \xrightarrow{\mathsf{byz\_send}\ sb\ r\ v} \langle F_s, F_r, F'_b, M' \rangle} \\
            \infer*[left=receive]
            {r \notin F_r \wedge (r \notin \mathsf{sender}(c) \vee r \in F_s) \wedge F_r' = F_r \cup \{r\} \wedge \Vec{v} \in \netsim_{\mathsf{\LLL}}\ M(r) }
            {\langle F_s, F_r, F_b, M \rangle \xrightarrow{\mathsf{receive}\ r\ \Vec{v}} \langle F_s, F_r', F_b, M \rangle} \\
        \end{mathpar}
    \end{definition}

    \LLL channel's state and transitions are mostly about bookkeeping of which nodes have performed their send and receive actions to prevent duplicates:
    each non-Byzantine sender node may only broadcast once;
    each receiver node may only receive once;
    and each Byzantine sender node can only send to each receiver node at most once.
    In addition, if the sender role is the same as the receiver role, which means a sender node is also a receiver node, it must send before it receives.

    In the $\mathsf{receive}$ transition rule, we account for network-reordering and dropping effects with $\netsim_\mathsf{\LLL}$.
    The initial state of the channel is the tuple $\langle \emptyset, \emptyset, \lambda\ \texttt{\_}.\emptyset, \lambda\ \texttt{\_}.[]\rangle$.
    
    \subsection{Composing Labeled Transition Systems}
    \label{sec:composing}

    We build a \emph{global} transition system from a set of local \LLL node and channel transition systems using the following notions of interaction composition:
    \begin{definition}[Interaction Composition]
    \label{def:interaction_composition}
    Let $\mathcal{T}_1 = \langle S_1, \Lambda_1, \rightarrow_1 \rangle$ and $\mathcal{T}_2 = \langle S_2, \Lambda_2, \rightarrow_2 \rangle$ be labeled transition systems, 
    and let $\Lambda \subseteq (\Lambda_1 + \unitt) \times (\Lambda_2 + \unitt)$.
    Then the interaction composition $\mathcal{T}_1 \bowtie_{\Lambda} \mathcal{T}_2$ is defined to be
    $\langle S_1 \times S_2, \Lambda, \rightarrow \rangle$
    where $(s_1, s_2) \xrightarrow{(\ell_1,\ell_2)} (s_1',s_2')$ when
    (a) $\ell_1 = \mathtt{tt}$ and $s_1' = s_1$ or $s_1 \xrightarrow{\ell_1}_1 s_1'$, and 
    (b) $\ell_2 = \mathtt{tt}$ and $s_2' = s_2$ or $s_2 \xrightarrow{\ell_2}_2 s_2'$.
    \end{definition}
    Note that we must specify a subset of permissible labels via $\Lambda$.  
    In other words, the choice of $\Lambda$ delimits the permissible combinations of local labels and thus can be used to model the interactions between the local systems.
    For instance, consider a system composed of a sender and a receiver.
    The only label for the sender is $\mathsf{send}\ x$, for any integer $x$.
    The only label for the receiver is $\mathsf{receive}\ y$, for any integer $y$.
    By only allowing the global labels to be of the form $(\mathsf{send}\ z, \mathsf{receive}\ z)$ for the same integer $z$, we can model synchronous communication between the sender and the receiver.
    
    Given the definition of interaction composition, a special case is when the sub-systems do not interact:
    \begin{definition}[Non-Interacting Composition]
    \label{def:non_interacting_composition}
    Let $\mathcal{T}_1 = \langle S_1, \Lambda_1, \rightarrow_1 \rangle$ and $\mathcal{T}_2 = \langle S_2, \Lambda_2, \rightarrow_2 \rangle$ be labeled transition systems.  Then the non-interacting composition $\mathcal{T}_1 \otimes \mathcal{T}_2$ is given by $\mathcal{T}_1 \bowtie_{\Lambda} \mathcal{T}_2$, where
    the labels in $\Lambda$ are either of the form $(\ell_1,\mathtt{tt})$ or $(\mathtt{tt},\ell_2)$.  
    \end{definition} 
    In other words, the non-interacting composition only allows local transitions.  
    When we have an indexed sequence of transition systems $[\mathcal{T}_1,\ldots,\mathcal{T}_n]$, we write $\otimes_i \mathcal{T}_i$ to 
    represent $\mathcal{T}_1 \otimes \cdots \otimes \mathcal{T}_n$.  If $s = (s_1,\ldots,s_n)$ is a state in a composed system, 
    we write $s@i$ to represent $s_i$.  Similarly, if $\ell$ is a label in the product, then $\ell@i$ represents the label's $i^{th}$ component.  

    \begin{definition}[Global Compilation]
    \label{def:global_compilation}
    Given a well-typed program $\tprog{\Delta}{\emptyset}{\mathcal{R}}{p}{\tau}$, 
    and a configuration $\mathcal{C}$, we define the \emph{global} compilation of $p$ to be:
    \[
        \compile{p}{} = (\otimes_{\Rrole, i \in \Rrole}\compile{p}{\Rrole,i}) \bowtie_{\Lambda} (\otimes_{c \in \Dom{\Delta}} \compile{c}{})
    \]
    where $\compile{c}{}$ for a channel $c$ is the initial channel state of the transition system described in Section~\ref{sec:channel}, 
    and where we write $i \in \Rrole$ to mean $i$ is the node identifier in
    the set of good nodes for role $\Rrole$ specified by configuration $\mathcal{C}$. 
    Additionally, we must specify $\Lambda$, the valid global labels for the composed system.  A
    global label $\ell$ must be of one of the three forms:
    \begin{itemize}
    \item Send:  for some $c$, $p$, and $v$, $\ell@c = \mathsf{send}\ p\ v$, $\ell@p = \mathsf{send}\ c\ v$, and for all $j$, $j \neq p \wedge j \neq c$, $\ell@j = \mathtt{tt}$.
    \item Byzantine:  for some $c$, $p$, $q$, and $v$, $\ell@c = \mathsf{byz\_send}\ p \ q\ v$, and for all $j \neq c$, $\ell@j = \mathtt{tt}$.
    \item Receive:  for some $c$, $p$, and $\Vec{v}$,  $\ell@p = \mathsf{receive}\ c\ \Vec{v}$ and $\ell@c = \mathsf{receive}\ p\ \Vec{v}$, 
        and for all $j$, $j \neq p \wedge j \neq c$, $\ell@j = \mathtt{tt}$.
    \end{itemize}
    \end{definition}
    This choice of global labels assumes that no more than one node and exactly one channel can take a local step in each discrete global step.
    The send label and receive label model the interaction between a single non-Byzantine node and a single channel.
    The two local labels must agree on the payload of the message(s) being sent/received.
    The $\mathsf{byz\_send}$ label models a stand-alone Byzantine action where an arbitrary message is sent from a Byzantine node to a non-Byzantine receiver through a channel.
    The asynchronous communication between two nodes is modeled in this way as two separate synchronous steps between a node and a channel.
    
    The state and the transition relation of this compiled system follow from the construction.  We also have a unique initial state $s_0$ for the global system:  
    \LLL node components map to their initial \LLL node program, while \LLL channel components map to the empty channel state.
    
    We say an \LLL system state $s_f$ is \emph{completed} if all its \LLL node instances have all been reduced to a $\mathsf{return}$.\footnote{Note that this may not be a terminal
    state in the transition system, as Byzantine sends could happen on channels.} We write $\extract{s}$ as the partial function which extracts these values and returns
    them as a record of vectors $\vec{v}$, such that $\vec{v}@i = v$ iff $s@i = \mathsf{return}\ v$ for any node $i$.  
    For an \LLL system instance compiled from a \HLL program, given any completed state $s_f$, we can turn $\extract{s_f}$ into a record
    value of type $\denote{\tau}$, where the different values are broken out by role.  
    We define $\extract{\compile{p}{}}$ to be the set of all such output values of the completed states reachable from the initial state.
    This is the set of possible outputs of the program $p$ in the \LLL (system) semantics.

    \subsubsection{Properties} A critical aspect of this construction is the relation between global behaviors (traces) and local views of behaviors.
    In particular, we make use of two key lemmas: a decomposition lemma that shows all global behaviors are necessarily arrangements of local behaviors, and a composition lemma that shows 
    we can always assemble traces from local systems into a global trace when there are suitable global labels.  
        
    We first define the projection from a sequence of global labels to sequences of local labels.
    \begin{definition}[Global Label Sequence Projections]
        \label{def:global_proj}
        Given a sequence of global labels $L \in (\Lambda_G)^*$ and a local system identifier $i$, the subsequence of global labels related to $i$ is defined as follows:
        \begin{equation*}
            \filterl{i}{L} \coloneqq
            \begin{cases}
                [] & L = [] \\
                \filterl{i}{L'} & L = l \cons L' \wedge l@i=\unitv \\
                l \cons \filterl{i}{L'} & L = l \cons L' \wedge l@i \ne \unitv \\                
            \end{cases}
        \end{equation*}

        The projection of $L$ to a sequence of $i$-local labels is defined as follows:
        \begin{equation*}
            \projl{i}{L} \coloneqq \mapl{(@i)}{\filterl{i}{L}}
        \end{equation*}
        
    \end{definition}
    Assume a system $\mathcal{T} = \mathcal{T}_1 \bowtie_{\Lambda} \mathcal{T}_2$, where $\mathcal{T}_i = \langle S_i, \Lambda_i, \rightarrow_i \rangle$.  

    \begin{lemma}[Decomposition]
        \label{thm:decomposition}
        Given a sequence of labels $L \in \Lambda^*$ and compound states $s$ and $t$ such that $s \xrightarrow{L} t$, 
        then $s@i \xrightarrow{\projl{i}{L}}_i t@i$.
    \end{lemma}

    \begin{lemma}[Composition]
        \label{thm:composition}
        Given component label sequences $L_i \in \Lambda_{i}^*$ and states $s_i, t_i$ such that $s_i \xrightarrow{L_i}_i t_i$, and
        a compound label sequence $L$ such that $\projl{i}{L} = L_i$, then $(s_1,s_2) \xrightarrow{L} (t_1,t_2)$. 
    \end{lemma}

    Lemma \ref{thm:decomposition} implies that any property proven for all possible behaviors of local systems is preserved by the global system.
    Intuitively, a local system is most ``free'' when the environment it is interacting with is arbitrary.
    Composing it with other systems provides more information about the environment and restricts its possible behaviors.

    Lemma \ref{thm:composition} shows that any combination of possible local system behaviors can happen in the global system if there exist suitable global labels.
    This can be used to prove a sequence of global labels are \emph{permissible} (\emph{i.e.,} lead from one global state to another) by establishing all its projections 
    are permissible local label sequences.

    \section{Adequacy Proof}
    \label{sec:proof}

    In this section, we sketch the proof that our denotational semantics is adequate for the compiled \LLL system operational semantics.

    For any closed \HLL program $\tprog{\Delta}{\emptyset}{\mathcal{R}}{p}{\tau}$, 
    $\denote{p}$ produces a set of possible outputs of type $\denote{\tau}$.
    The program also compiles to an \LLL system and yields a set of possible outputs of that system 
    $\extract{\compile{p}{}}$.
    \begin{theorem}[Adequacy]
        \label{thm:adequacy}
        $\extract{\compile{p}{}} \subseteq \denote{p}$. 
    \end{theorem}
    We prove the above result in two steps:
    
    \textit{Step 1:} By definition, any possible output of an \LLL system is witnessed by a trace $L$ such that there exists a completed state $s_f$ which gives the output and $s_0 \xrightarrow{L} s_f$.
    We show that we can reduce the set of traces to be considered in the operational semantics to an ``aligned'' subset that contains the same outputs and whose execution order puts all of the operations on a given channel together.  

    \textit{Step 2:} We prove $\extract{s_f} \in \denote{p}$ assuming $L$ is an aligned trace.

    \subsection{Reduction to Aligned Traces}

    We begin by defining the notion of the \textit{alignment} of a global trace with respect to a channel context.  Intuitively, this captures the idea that the outputs of 
    a given trace are the same as a trace that forces all of the interactions to happen in an order specified by $\Delta$.  
    
    \begin{definition}[Alignment of a Global Trace]
    Given a channel context $\Delta$ and a sequence of global \LLL system labels $L$, the alignment of $L$ with respect to $\Delta$ is given by:
        \begin{equation*}
            \alignl{\Delta}{L} \coloneqq \flatmapl{(\lambda c.\; \filterl{c}{L})}{\Dom{\Delta}}
        \end{equation*}      
    \end{definition}
    That is, $\alignl{\Delta}{L}$ takes a list of labels, and for each channel $c$, projects out the sub-list of labels corresponding to non-empty transitions for $c$, and then concatenates those lists back together in the order that the channels are listed in $\Delta$.  Note that for every global label, there is a unique channel that transitions.  Thus, $\alignl{\Delta}{L}$ produces a permutation of the labels in $L$.

    For example, consider a system with two good nodes 
    $n_1$ and $n_2$ with $n_1$ sending messages to $n_2$ through channel $c_1$ followed by $c_2$, and a Byzantine node $b$ sending a message on $c_1$.  
    One possible sequence of actions we could see is the following
    \[
    [n_1(\mathsf{send}\ c_1\ v_1) ; n_1(\mathsf{send}\ c_2\ v_2) ; n_2(\mathsf{receive}\ c_1\ \Vec{v_1}); n_2(\mathsf{receive}\ c_2\ \Vec{v_2}) ; b(\mathsf{byz\_send}\ c_1\ w)]
    \]
    where we only show the relevant node-portions of the labels.  
    After aligning the trace, we will have:
    \[
    [n_1(\mathsf{send}\ c_1\ v_1) ; n_2(\mathsf{receive}\ c_1\ \Vec{v_1}) ; b(\mathsf{byz\_send}\ c_1\ w); n_1(\mathsf{send}\ c_2\ v_2) ; n_2(\mathsf{receive}\ c_2\ \Vec{v_2})]
    \]
    
    So that all of the $c_1$ actions are performed before the $c_2$ actions, but otherwise the relative order of transitions is preserved.    

    \begin{theorem}[Aligned Trace Permissibility]
        \label{thm:aligned_permissibility}
        Let $s_0$ be the initial state of an \LLL system built from a well-typed \HLL program $p$ with the channel context $\Delta$ and suppose
        $s_0 \xrightarrow{L} s$.  Then $s_0 \xrightarrow{\alignl{\Delta}{L}} s$.  
    \end{theorem}
    
    \begin{proof}
        We use the composition lemma (Lemma~\ref{thm:composition}) to prove $\alignl{\Delta}{L}$ is a permissible trace of the system.
        This requires us to prove that any local projection of $\alignl{\Delta}{L}$ to a local system is a permissible trace of that local system.
        By the decomposition lemma (Lemma~\ref{thm:decomposition}), we know any projections of the given trace $L$ are permissible.
        So it suffices to prove that any local projection of $\alignl{\Delta}{L}$ is equivalent to that of $L$.
        
        In an \LLL system, we have two kinds of local systems: channels and nodes.
        Recall that $\alignl{\Delta}{L}$ reorders the labels in the channel order captured by $\Delta$, 
        so it preserves the projections to the channels naturally.
        More formally, for any channel $c$ in the system, we have:
        \begin{align*}
        \projl{c}{\alignl{\Delta}{L}} &= \mapl{(@c)}{(\filterl{c}{\flatmapl{(\lambda c'. \filterl{c'}{L})}{\Dom{\Delta}}})}
                              && \doublebackslash \text{unfold definitions} \\
                              &= \mapl{(@c)}{({\flatmapl{(\lambda c'. \filterl{c}{\filterl{c'}{L}})}{\Dom{\Delta}}})}
                              && \doublebackslash  \text{list properties} \\
                              &= \mapl{(@c)}{(\filterl{c}{L})} && \doublebackslash \text{other channels give []} \\
                              &= \projl{c}{L} && \doublebackslash \text{by definition}
        \end{align*}
        
        For nodes, 
        we begin by defining a relation ($\orderrelated_\Rrole$) that holds when a sequence of transition labels respects the
        order in $\Delta$ according to a given role $\Rrole$:
        \begin{align*}
            [] & \ \orderrelated_\Rrole\  \Delta \\
            (\mathsf{send}\ c\ v)::(\mathsf{receive}\ c\ \Vec{v})::L & \ \orderrelated_R\  (c: (\Rrole, \Rrole, \tau))::\Delta\ \mathsf{when}\ L \ \orderrelated_\Rrole\  \Delta \\
            (\mathsf{send}\ c\ v)::L & \ \orderrelated_R\  (c: (\Rrole,S,\tau))::\Delta \ \mathsf{when}\ L \ \orderrelated_\Rrole\  \Delta \\
            (\mathsf{receive}\ c\ \Vec{v})::L & \ \orderrelated_\Rrole\  (c:(S,\Rrole,\tau))::\Delta \ \mathsf{when}\ L \ \orderrelated_\Rrole\ \Delta \wedge \Rrole \neq S \\
            L & \ \orderrelated_\Rrole\  (c':(S_1,S_2,\tau))::\Delta \ \mathsf{when}\ L \ \orderrelated_\Rrole\ \Delta \wedge \Rrole \neq S_1 \wedge R \neq S_2
        \end{align*}
        We then extend this relation to describe when an \LLL node program respects $\Delta$.
        Our intuition is that all of the possible traces the program can generate should respect the ordering in $\Delta$:
        \begin{align*}
            \mathsf{return}\ v & \ \orderrelated_\Rrole\  \Delta \ \mathsf{when}\ \Rrole \notin \Delta \\
            \mathsf{sendThen}\ c\ v\ (\mathsf{rcvThen}\ c\ k) & \ \orderrelated_\Rrole\  (c:(\Rrole,\Rrole,\tau))::\Delta\ \mathsf{when}\ \forall \Vec{v}, k(\Vec{v}) \ \orderrelated_\Rrole\  \Delta \\
            \mathsf{sendThen}\ c\ v\ k & \ \orderrelated_\Rrole\  (c:(\Rrole,S,\tau))::\Delta\ \mathsf{when}\ k \ \orderrelated_\Rrole\  \Delta \wedge \Rrole \neq S \\
            \mathsf{rcvThen}\ c\ k & \ \orderrelated_\Rrole\  (c:(S,\Rrole,\tau))::\Delta\ \mathsf{when}\ \forall \Vec{v}, k(\vec{v}) \ \orderrelated_\Rrole\  \Delta\\
            t &\ \orderrelated_\Rrole\  (c'\hastype(S_1,S_2,\tau))::\Delta \ \mathsf{when}\ t \ \orderrelated_\Rrole\ \Delta \wedge \Rrole \neq S_1 \wedge \Rrole \neq S_2
        \end{align*}

        \begin{lemma}[Compilation Respects Channel Ordering]
            \label{lem:order_preservation}
            Suppose $\tprog{\Delta}{\emptyset}{\mathcal{R}}{p}{\tau}$, and let $\Rrole \in \mathcal{R}, i \in \Rrole$.  
            Let $L$ be a sequence of \LLL node labels and $t$ an \LLL node program such that $\compile{P}{\Rrole,i} \xrightarrow{L} t$.  Then $L \ \orderrelated_\Rrole\  \Delta.$
        \end{lemma}
     
        \begin{proof}
            It is easy to see by induction on the typing derivation for $p$ that $\compile{p}{\Rrole,i} \ \orderrelated_\Rrole\  \Delta$.
            We argue that for any \LLL node program $t$ such that $t \ \orderrelated_\Rrole\  \Delta$, that if $t \xrightarrow{L} t'$, then $L \ \orderrelated_\Rrole\  \Delta$, and
            furthermore, there exists a suffix of $\Delta$, $\Delta'$ such that $t' \ \orderrelated_\Rrole\  \Delta'$.
            The argument proceeds by induction on the length of $L$ and then via case analysis on the structure of $t$.  
        \end{proof}

    Now for any role $\Rrole$ and non-Byzantine node $i$ in that role, let $L_i = \projl{i}{L}$.  By Lemma~\ref{lem:order_preservation}, we know that $L_i \ \orderrelated_\Rrole\  \Delta$ 
    and we need to show $L_i = \projl{i}{\alignl{\Delta}{L}}$.  The proof proceeds by induction on the derivation of $L_i \ \orderrelated_\Rrole\  \Delta$.  
    \end{proof}

    \subsection{Adequacy of Aligned Traces}

    We now prove Theorem ~\ref{thm:adequacy} by showing that the output of any completed and aligned trace belongs to the set of outputs given by the denotational semantics. 

    Formally, we show for any closed \HLL program $P$, well-typed with $\tprog{\Delta}{\Gamma}{\mathcal{R}}{P}{\tau}$,
    \[
        \extract{\compile{P}{}} \subseteq \denote{P}
    \]

    By definition of $\extract{}$, for any output $o \in \extract{\compile{P}{}}$, there exists a trace $l$ and a state $F$ such that $F$ is a complete state, $\extract{F} = o$, and $\compile{P}{} \xrightarrow{l} F$.

    By the alignment theorem, we can assume $l = \alignl{\Delta}{l}$ without loss of generality.

    Thus, assume $\Delta = [c_1, c_2, \dots, c_n]$, $l = l_1 \doubleplus l_2 \doubleplus \dots l_n$, where $l_i$ contains only symbols associated with channel $c_i$.

    Additionally, we can assume $P$ is in the let-normal form:
    \begin{align*}
        P = \ & \letexp{x_1}{\communicate{c_1}{m_1}{d_1}{f_1}}{} \\
              & \letexp{x_2}{\communicate{c_2}{m_2}{d_2}{f_2}}{} \\
              & \dots \\
              & \letexp{x_n}{\communicate{c_n}{m_n}{d_n}{f_n}}{} \\
              & \mathtt{ret}\ \{ R_i \mapsto e_i \}
    \end{align*}

    We now define a big-step operational semantics for any closed, well-typed, and normalized \HLL program $P$ with configuration $\mathcal{C}$.
    \begin{mathpar}
        \infer*
        { \Delta(c) = (S, R, \tau_m) \quad  L = [l_1, l_2, \dots, l_{n_R - b_R}] \\ \forall i, l_i \in \netsim_S(\denote{m}) \quad y = \{ R \mapsto \mathsf{map3} \ \mathsf{foldl}\ \denote{f}\ \denote{d}\ L \}}
        { \letexp{x}{\communicate{c}{m}{d}{f}}{P'} \xrightarrow{L} P'[y/x]}
    \end{mathpar}

    We use $\extract{\langle P \rangle}$ to denote the set of possible outputs in this big-step semantics, defined as follows:
    \begin{align*}
        \extract{\langle \mathtt{ret}\ \{ R_i \mapsto e_i \} \rangle} & = \mathsf{singleton} \{R_i \mapsto \denote{e_i} \} \\ 
        \extract{\langle \letexp{x}{\communicate{c}{m}{d}{f}}{P'_x} \rangle} & = \bigcup_{P \xrightarrow{L} P'} \extract{\langle P' \rangle} 
    \end{align*}

    This splits the proof of the adequacy theorem into two parts $ \extract{\compile{P}{}} \subseteq \extract{\langle P \rangle} $ and $\extract{\langle P \rangle} \subseteq \denote{P} $.

    \subsubsection{Big-step to Denotational}

    For $\extract{\langle P \rangle} \subseteq \denote{P} $, we perform an induction on the structure of $P$.

    \begin{proof}
        Base case: when $P = \mathsf{ret} \{ R_i \mapsto e_i \}$, $\extract{\langle P \rangle} =\mathsf{singleton} \{R_i \mapsto \denote{e_i} \} = \denote{P}$.

        Inductive case: when $P = \letexp{x}{\communicate{c}{m}{d}{f}}{P'_x}$. By definition of the big-step semantics, $\extract{\langle P \rangle} = \bigcup_{P \xrightarrow{L} P'} \extract{\langle P' \rangle}$. It suffices to prove for any $L$ and $P'$ such that $P \xrightarrow{L} P'$, $\extract{\langle P' \rangle} \subseteq \denote{P}$.
        
        By the induction hypothesis, we know $\extract{\langle P' \rangle} \subseteq \denote{P'}$. It is enough to prove $\denote{P'} \subseteq \denote{P}$.

        This follows from the definition of the big-step semantics and the denotational semantics.
    \end{proof}

    \subsubsection{\LLL to Big-step}

    To move from \LLL traces to Big-step traces, we first need a well-formedness property of \LLL traces.

    \begin{definition}[Finished channel]
        Recall that a state of an \LLL channel $c$ is a tuple $\langle F_s, F_r, F_b, M \rangle$, where $F_s$ is the set of non-Byzantine senders and $F_r$ is the set of non-Byzantine receivers.

        A state is \emph{finished}, if and only if $F_s = \mathsf{sender}(c)$ and $F_r = \mathsf{receiver}(c)$. This means each non-Byzantine sender has performed their send operation, and each non-Byzantine receiver has performed their receive operation.
    \end{definition}

    The first lemma we need here is that all channels are finished in a complete \LLL trace.

    \textbf{Proof Sketch:} In a complete trace, each node must be complete. Because the compilation guarantees each node respects the channel context, $\Delta$, for each channel, each of the non-Byzantine senders and receivers must have performed their send and receive actions on this channel. Thus, every channel is finished.

    For convenience, we extend the definition of \LLL channel state with two extra bookkeeping states. $M_s \hastype \mathsf{list}\ \tau_m$ records each message sent by each non-Byzantine sender node. $M_r \hastype F_r \funarrow \mathsf{list}\ \tau_m$ records the list of messages received by each non-Byzantine receiver node.
    We modify the transitions accordingly to put information into those two pieces of the states.

    Importantly, with the extra information, we can extract a big-step label $L = [l_1, l_2, \dots, l_{n_R - b_R}]$ from a finished channel state by having $l_i = M^f_r(r_i)$, where $M^f_r$ is the $M_r$ of the finished state and $r_i$ is a specific node of the receiver role.

    So the second lemma we need is that this always gives a valid $L$ such that $\forall i, l_i \in \netsim_S(M^f_s)$. $M^f_s$ is the $M_s$ of the finished state and contains the messages sent by all non-Byzantine sender nodes.

    \textbf{Proof Sketch:} We do an induction over the \LLL channel trace. By definition, each received list of messages satisfies $\netsim_{\LLL}$ with the list of messages that have been sent to the receiver at the time of the receive action. We need to prove that each received list of messages satisfies $\netsim_S(M^f_s)$. This can be proven with the invariant that $M(r)$ is always a subset of $\mathsf{add\_any}\, M^f_s\, b_R$ and thus $\netsim_{\LLL}(M(r)) \subseteq \netsim_S(M^f_s)$.

    A third lemma is needed to tie the changes in the node state to the actions that happened in the channel.
    More specifically, a node that sends a message will always send the message described by its program and continue with a unit; a node that receives a list of messages will continue with the result of folding its handler over the list of received messages.
    This lemma can be proven following the definitions of the \LLL transitions.

    The last lemma allows us to erase finished channels from the \LLL state. For any \LLL trace $S \xrightarrow{l} F$ and $S$ contains a finished channel state $S@c$, then $S' \xrightarrow{l'} F'$, where $S'$ is $S$ except $S@c$, $F'$ is $F$ except $F@c$, and $l'$ is $l$ without labels associated with $c$.
    This lemma holds because the only transitions in $l$ associated with $c$ are Byzantine send operations on channel $c$ due to $S@c$ being in a finished state and these operations can only affect the state of $c$.
    
    To prove $ \extract{\compile{P}{}} \subseteq \extract{\langle P \rangle} $, we perform an induction on the structure of $P$.

    \begin{proof}
        Base case: when $P = \mathsf{ret} \{ R_i \mapsto e_i \}$, by the definition of the compilation rules, both sides are the singleton set of $ \{R_i \mapsto \denote{e_i} \} $.

        Inductive case: when $P = \letexp{x}{\communicate{c}{m}{d}{f}}{P'_x}$. It suffices to show, for any $\Delta${-}aligned \LLL trace $\compile{P}{} \xrightarrow{l} F$, $\extract{F} \in \extract{\langle P \rangle}$.

        We consider a block of labels of the same channel at a time.
        Because $l$ is $\Delta${-}aligned, it must start with a block of labels associated with $c$, so $l = l_c \doubleplus l'$ and $\compile{P}{} \xrightarrow{l_c} S \xrightarrow{l'} F$.
        By the first lemma above, because $F$ is a completed state, it contains a finished channel $c$ state $F@c$.
        By the definition of alignment, $l'$ contains no labels associated with $c$, so $S@c = F@c$ and $S@c$ is a finished channel state.
        By the second lemma above, we can extract a valid big-step label $L$ from $S@c$.
        This gives us a new \HLL program $P'$ through $P \xrightarrow{L} P'$.
        By the third lemma above, the node states in $S$ match those of $\compile{P'}{}$. They only differ because $P'$ does not have channel $c$.
        And by the fourth lemma, there exists $F'$, such that $\compile{P'}{} \xrightarrow{l'} F'$.
        We know $\extract{F'} = \extract{F}$ because the only difference between $F$ and $F'$ is that $F$ includes channel $c$'s state and $F'$ does not, and the extracted output does not depend on the channel states.
        By the induction hypothesis, $\extract{F'} \in \extract{\langle P' \rangle}$.
        By the definition of the big-step semantics, $\extract{\langle P' \rangle} \subseteq \extract{\langle P \rangle}$.
        So we conclude $\extract{F} \in \extract{\langle P \rangle}$.
    \end{proof}

    \section{Related Work}

    Despite our goal being to explore new directions for formal reasoning of distributed systems rather than pushing the state-of-the-art of what can be formally verified, our work is closely related to formal verification frameworks of distributed systems.

    \textit{Reducing Asynchrony to Synchrony.}
    The line of work on the Heard-Of model ~\cite{charron-bostHeardOfModelComputing2009, charron-bostFormalVerificationConsensus2011, dragoiPSyncPartiallySynchronous2016a, damianCommunicationClosedAsynchronousProtocols2019} also observes that certain asynchronous protocols can be verified by reasoning about their synchronous counterparts. 
    \cite{charron-bostFormalVerificationConsensus2011} formalizes the model in Isabelle/HOL.
    PSync~\cite{dragoiPSyncPartiallySynchronous2016a} is embedded in Scala and provides a high-level, synchronous language for describing crash fault-tolerant systems, which can be executed on partially asynchronous networks with a runtime. 
    \cite{damianCommunicationClosedAsynchronousProtocols2019} proposes and formalizes the concept of \emph{communication-closed} protocols, in which communication happens in rounds and messages are either received in that round or not at all. It then generalizes the asynchrony to synchrony reduction to such protocols.
    
    While we offer a functional semantics for compositional theorem-proving in Rocq, the Heard-Of model is mostly adopted for automated verification of manually annotated imperative programs.
    Furthermore, as a language-based approach, well-typed \HLL programs are adequate for reasoning by construction, while in ~\cite{damianCommunicationClosedAsynchronousProtocols2019}, the communication-closed restrictions are enforced by manually annotating C programs and then checked by a stand-alone checker tool.

    \textit{Modeling and Reasoning about Byzantine Failures.}
    We model both crash failures and Byzantine failures through the unified $\netsim$ abstraction.
    Velisarios ~\cite{rahliVelisariosByzantineFaultTolerant2018} and its successor, Asphalion ~\cite{vukoticAsphalionTrustworthyShielding2019}, are verification frameworks built in Rocq that also support reasoning about Byzantine behaviors.
    Using a state machine model and the logic of events ~\cite{bickfordLogicEventsFramework2012}, the protocols need to be reasoned about at a relatively low level similar to that of \LLL.
    However, the frameworks provide epistemic models for the Byzantine adversary and built-in primitives for cryptographic signatures.
    The authors used the framework to verify and produce an executable artifact of the PBFT protocol ~\cite{castroPracticalByzantineFault1999}.

    \textit{Layered Refinements.}
    An enabling technique for many existing works that use trace-based semantics is \emph{layered refinement}. 
    Layered refinements have been shown to be effective at reducing the complexity of those proofs.
    LiDODAG~\cite{qiuLiDODAGFrameworkVerifying2025} and its predecessors~\cite {qiuLiDOLinearizableByzantine2024, honoreMuchADOFailures2021, honoreAdoreAtomicDistributed2022} are Rocq-based verification frameworks for distributed systems that provide several state machine-based trace semantics at different abstraction levels.
    For instance, the top-level trace models the runtime behavior operationally as a non-deterministically growing tree, while at lower levels, the system state contains more details, such as program counters for each node and network messages.
    The authors verified a handful of consensus protocols with more advanced features, such as DAG-based consensus and dynamic reconfiguration, using their layered semantics by modeling it at those different levels and proving refinement relations between them.
    Our compilation from \HLL to \LLL has a similar flavor to that of layer refinement, with two differences:
    (1) The relationship between \HLL and \LLL is proven once and for all, instead of for every verification instance; and (2) \HLL's semantics is functional and can be decomposed along its syntax, while layered refinement can only move from one state transition system to another.
    
    \textit{Separation Logic.}
    Distributed systems that assume more benign execution environments can adopt off-the-shelf techniques such as separation logic for reasoning and provide richer node-local semantics.
    Aneris~\cite{krogh-jespersenAnerisMechanisedLogic2020} is a Rocq framework that utilizes the concurrent separation logic framework Iris~\cite{jungIrisMonoidsInvariants2015} 
    to provide local reasoning for the nodes involved. It supports higher-order store and network sockets and has been used to verify a load balancer with concurrent local programs.
    Disel~\cite{sergeyProgrammingProvingDistributed2017} is also a separation logic-based framework in Rocq and supports compositional reasoning about interactions between an abstract core distributed system and multiple clients.
    Grove~\cite{sharmaGroveSeparationLogicLibrary2023} is another concurrent separation logic library in Rocq. It supports many system features, including time-based leases, reconfiguration, crash recovery, and thread-level concurrency.

    \textit{Decidable Logics.}
    Another approach that eases the proof burden is to restrict the implementation and the specification such that the verification problem falls under some decidable logic and can be solved automatically, as shown in Ivy~\cite{taubeModularityDecidabilityDeductive2018}, TIP~\cite{berkovitsVerificationThresholdBasedDistributed2019}, and also adopted in Verus~\cite{lattuadaVerusPracticalFoundation2024}.
    The design of these languages explores a direction different from ours, where a high degree of automation can be achieved through SMT or specialized solvers.
    
    \bibliography{main}
    
\end{document}